\documentclass[10pt,journal,twocolumn]{IEEEtran}
\usepackage{amsmath}%
\usepackage{amsfonts}%
\usepackage{amssymb}%
\usepackage{graphicx}%
\usepackage{array}%
\usepackage{subfig}%

\newtheorem{theorem}{Theorem}

\newtheorem{corollary}{Corollary}

\newtheorem{lemma}{Lemma}

\newtheorem{proposition}{Proposition}

\newenvironment{proof}[1][Proof]{\textbf{#1.} }{\ \rule{0.5em}{0.5em}}
\DeclareMathOperator{\snr}{SNR}
\begin{document}

\title{On Rate-Splitting by a Secondary Link in Multiple Access Primary Network} 

\author{
\IEEEauthorblockN{John Tadrous and Mohammed Nafie}
\thanks{Authors are with the Wireless Intelligent Networks Center (WINC), Nile University, Cairo, Egypt.\newline \indent E-mail: john.tadrous@nileu.edu.eg, mnafie@nileuniversity.edu.eg}}
\maketitle

\begin{abstract}
An achievable rate region is obtained for a primary multiple access network coexisting with a secondary link of one transmitter and a corresponding receiver. The rate region depicts the sum primary rate versus the secondary rate and is established assuming that the secondary link performs rate-splitting. The achievable rate region is the union of two types of achievable rate regions. The first type is a rate region established assuming that the secondary receiver cannot decode any primary signal, whereas the second is established assuming that the secondary receiver can decode the signal of one primary receiver. The achievable rate region is determined first assuming discrete memoryless channel (DMC) then the results are applied to a Gaussian channel. In the Gaussian channel, the performance of rate-splitting is characterized for the two types of rate regions. Moreover, a necessary and sufficient condition to determine which primary signal that the secondary receiver can decode without degrading the range of primary achievable sum rates is provided. When this condition is satisfied by a certain primary user, the secondary receiver can decode its signal and achieve larger rates without reducing the primary achievable sum rates from the case in which it does not decode any primary signal. It is also shown that, the probability of having at least one primary user satisfying this condition grows with the primary signal to noise ratio.

\end{abstract}


\section{Introduction}
\IEEEPARstart{A} {potential} benefit of allowing secondary users to share primary bands is the enhancement of the spectrum utilization. As introduced in \cite{Mitola} and \cite{Ian}, cognitive radios, or secondary users, are frequency agile devices that can utilize unused spectrum bands through dynamic spectrum access. In dynamic spectrum access secondary users should sense the spectrum and identify unused bands, or spectrum holes. If a band is sensed and found to be in low use by primary users, i.e., underutilized, a secondary user may opportunistically access this band by adjusting its transmit parameters to fully utilize this band without causing excessive interference on the primary users. However, a secondary user has to leave this band and switch to another if the demand by primary users increases.

The notion of dynamic spectrum access has opened research in different problems regarding the new functionalities that a secondary user should perform, e.g., spectrum sensing, spectrum sharing, spectrum mobility and spectrum management \cite{Ian} and \cite{Gridlock}. Moreover, information theoretic bounds on potential achievable rates by cognitive radio networks are being investigated. In most of those works cooperation between primary and secondary transmitters is considered. In \cite{Achievable} an achievable rate region of primary versus secondary users' rates is introduced when a cognitive transmitter has full knowledge of the primary message in a two-transmitter two-receiver interference channel and the primary user cooperates with the secondary link through rate-splitting introduced in \cite{HanKob}. In \cite{Partial} and \cite{Full} the notion of conferencing is introduced for the interference channel where the cognitive link is assumed to know part or all of the message of the primary transmitter.

In this paper we consider a multiple access channel (MAC) of two transmitters and a common receiver shared by a secondary link of single transmitter and a corresponding receiver. The secondary transmitter is assumed to employ rate-splitting by dividing its signal into two parts: one part is decodable by the secondary receiver and treated as noise by the primary receiver, whereas the other part is decodable at both receivers. Based on this scheme we:
\begin{itemize}
	\item Establish an achievable rate region, $\mathcal{R}^o$, for the primary sum rate versus the secondary rate in a discrete memoryless channel (DMC) setup assuming that all of the primary signals are treated as noise at the secondary receiver.
	\item Establish another achievable rate region, $\mathcal{R}^r_i$, for which the signal of primary transmitter $i$ is to be fully decodable at the secondary receiver besides being decodable at the primary receiver. For this scheme we show that there exists a case for which $\mathcal{R}^r_i$ includes $\mathcal{R}^o$.
	\item Provide an overall achievable rate region $$\mathcal{R}=\mathcal{R}^o\bigcup\left(\cup_{i\in\{1,2\}}\mathcal{R}_i^r\right).$$ 
	\item Apply the results obtained in DMC case in a Gaussian setup where the effect of rate-splitting on the achievable rate region is analyzed. A necessary and sufficient condition is established for obtaining the overall rate region without rate-splitting.  
	\item Derive a necessary and sufficient condition so that the secondary receiver can decode the signal of one primary user without affecting the range of achievable primary sum rates, but only enhances the range of achievable secondary rates. We call this condition \emph{primary decodability condition for Gaussian (PDCG)}  channel.
	\item Show, numerically, that the probability of having at least one primary user satisfying PDCG monotonically increases with the signal-noise-ratio of the primary users.  
\end{itemize}
We have provided some of the results in this paper in a conference paper version \cite{ICT}.
The introduced network model of MAC primary network shared by secondary operations has been addressed in some resource allocation frameworks without rate-splitting by secondary users \cite{Chandramouli}-\cite{Globecom}. Rate-splitting by a secondary link, however, has been introduced in \cite{Opportunistic} where the secondary user is assumed to know the codebook of a primary transmitter and opportunistically splits its rate into two parts and decodes it in the following way. It decodes the first part treating both the primary signal and the second part as noise, decodes and cancels the primary signal and then decodes the second part. This scheme is generalized in this paper as we consider the cases when the signal of one primary transmitter is decodable at the secondary receiver and when all the primary signals are treated as noise.


The rest of this paper is organized as follows. In Section \ref{sec:Pre} the discrete memoryless channel (DMC) models are defined. In Section \ref{sec:ARR} the achievable rate regions are established for the defined DMC models. Then, obtained results are applied in a Gaussian channel setup in Section \ref{sec:Gauss} and the paper is conncluded in Section \ref{sec:conc}.  

\section{Channel Model}
\label{sec:Pre}
In our formulation we denote random variables by $X$, $Y$, $\cdots$ with realizations $x$, $y$, $\cdots$ from sets $\mathcal{X}$, $\mathcal{Y}$, $\cdots$ respectively. The communication channel is considered to be discrete and memoryless.
\subsection{Basic Channel Model}
We consider a basic channel $C_B$ defined by a tuple $(\mathcal{X}_1,\mathcal{X}_2,\mathcal{X}_s,\omega,\mathcal{Y}_p,\mathcal{Y}_s)$, where $\mathcal{X}_1$, $\mathcal{X}_2$ are two finite input alphabet sets of the primary transmitters and $\mathcal{X}_s$ is a finite input alphabet set of the secondary transmitter. Sets $\mathcal{Y}_p$ and $\mathcal{Y}_s$ are two finite output alphabet sets at the primary and secondary receivers respectively, and $\omega$ is a collection of conditional channel probabilities $\omega(y_py_s|x_1x_2x_s)$ of $(y_p,y_s)\in{\mathcal{Y}_p\times{\mathcal{Y}_s}}$ given $(x_1,x_2,x_s)\in{\mathcal{X}_1\times\mathcal{X}_2\times\mathcal{X}_s}$, with marginal conditional distributions:
\begin{eqnarray*}
\omega_p(y_p|x_1x_2x_s)=\sum_{y_s\in \mathcal{Y}_s}{\omega(y_py_s|x_1x_2x_s)}, \\
\omega_s(y_s|x_1x_2x_s)=\sum_{y_p\in \mathcal{Y}_p}{\omega(y_py_s|x_1x_2x_s)}.
\end{eqnarray*}
Since the channel is memoryless, the conditional probability $\omega^n(\textbf{y}_p\textbf{y}_s|\textbf{x}_1\textbf{x}_2\textbf{x}_s)$ is given by
\begin{equation*}
\label{eq:cond_prob}
\omega^n(\textbf{y}_p\textbf{y}_s|\textbf{x}_1\textbf{x}_2\textbf{x}_s)=\prod_{t=1}^{n}{\omega(y_p^{(t)}y_s^{(t)}|x_1^{(t)}x_2^{(t)}x_s^{(t)})},
\end{equation*}
where
\begin{eqnarray*}
\label{eq:definitions}
		\textbf{x}_a=&(x_a^{(1)},\cdots,x_a^{(n)})\in{\mathcal{X}_a^n},&a=1,2,s,\\	
		\textbf{y}_a=&(y_a^{(1)},\cdots,y_a^{(n)})\in{\mathcal{Y}_a^n},& a=p,s.
\end{eqnarray*}
The same also holds for the marginal conditional distributions $\omega^n_p(\textbf{y}_p|\textbf{x}_1\textbf{x}_2\textbf{x}_s)$ and $\omega^n_s(\textbf{y}_s|\textbf{x}_1\textbf{x}_2\textbf{x}_s)$.
Let $\mathcal{M}_1=\{1,\cdots,M_1\}$, $\mathcal{M}_2=\{1,\cdots,M_2\}$ be message sets for primary transmitters 1 and 2 respectively, and $\mathcal {M}_s=\{1,\cdots,M_s\}$ be a message set for the secondary transmitter. A code $(n,M_1,M_2,M_s,\epsilon)$ is a collection of $M_1$, $M_2$ and $M_s$ codewords such that:
\begin{enumerate}
	\item Sender $a$, $a=1,2,s$, has an encoding function $\phi_a: i\rightarrow \textbf{x}_{ai}$, $i\in\mathcal{M}_a$ and $\textbf{x}_{ai}\in\mathcal{X}^n$.
	\item The primary receiver has $M_1M_2$ disjoint decoding sets $\mathcal{D}_{pij}\subseteq{\mathcal{Y}_p^n}$, $ij\in\mathcal{M}_1\times\mathcal{M}_2$, and a decoding function $\psi_p: \textbf{y}_p\rightarrow ij$ if $\textbf{y}_p\in\mathcal{D}_{pij}$, where $ij\in\mathcal{M}_1\times\mathcal{M}_2$.
	\item The secondary receiver has $M_s$ disjoint decoding sets $\mathcal{D}_{sk}\subseteq{\mathcal{Y}_s^n}$, $k\in\mathcal{M}_s$, and a decoding function $\psi_s: \textbf{y}_s\rightarrow{k}$ if $\textbf{y}_s\in\mathcal{D}_{sk}$, where $k\in\mathcal{M}_s$ (see Fig.\ref{fig:BModel}).
	\item Probability of error for the primary network and the secondary link are less than $\epsilon$, that is, $Pe_p\leq{\epsilon}$ and $Pe_s\leq{\epsilon}$ respectively, where
	\begin{eqnarray}
\label{eq:EprobCBp}	Pe_p=\frac{1}{M_1M_2M_s}\sum_{i,j,k}\omega_p^n(\textbf{y}_p\notin\mathcal{D}_{pij}|\textbf{x}_{1i}\textbf{x}_{2j}\textbf{x}_{sk}),\\
\label{eq:EprobCBs}
Pe_s=\frac{1}{M_1M_2M_s}\sum_{i,j,k}\omega_s^n(\textbf{y}_s\notin\mathcal{D}_{sk}|\textbf{x}_{1i}\textbf{x}_{2j}\textbf{x}_{sk}).
	\end{eqnarray}
\end{enumerate}
\begin{figure}
	\centering
		\includegraphics[width=0.45\textwidth]{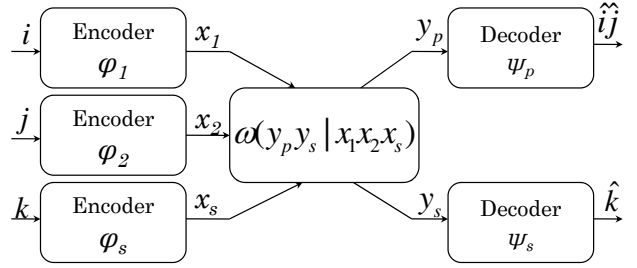}
	\caption{Basic channel model $C_B$}
	\label{fig:BModel}
\end{figure}

A rate tuple $(R_1,R_2,R_s)$ of nonnegative real values is achievable if for any $\eta>0$, $0<\epsilon<1$ there exists a code such that
\begin{equation}
\label{eq:R_Bineq}
\frac{1}{n}\log{M_a}\geq{R_a-\eta}, \   \ a=1,2,s,
\end{equation} 
with sufficiently large $n$.
\subsection{Rate-Splitting Channel}
Rate-splitting channel, $C_{RS}$, is a modified version of the basic channel $C_B$, where $C_{RS}$ is defined by a tuple $(\mathcal{X}_1,\mathcal{X}_2,\mathcal{X}_s,\omega,\mathcal{Y}_p,\mathcal{Y}_s)$ with its elements are as defined in $C_B$. Moreover, the input message sets for the primary transmitters are also $\mathcal{M}_1$ and $\mathcal{M}_2$ exactly as in $C_B$. However, the secondary user is assumed to have two finite message sets $\mathcal{L}_s=\{1,\cdots,L_s\}$, $\mathcal{N}_s=\{1,\cdots,N_s\}$. Hence, a code $(n,M_1,M_2,L_s,N_s,\epsilon)$ over the channel $C_{RS}$ is a collection of $M_1$, $M_2$, $L_sN_s$ codewords such that:
\begin{enumerate}
	\item Primary transmitter $a$, $a=1,2$, has an encoding function $\phi_a: i\rightarrow \textbf{x}_{ai}$, $i\in\mathcal{M}_a$, $\textbf{x}_{ai}\in\mathcal{X}_a^n$.
	\item The secondary transmitter has an encoding function $\phi_s: kl\rightarrow\textbf{x}_{skl}$, $kl\in\mathcal{L}_s\times\mathcal{N}_s$, $\textbf{x}_{skl}\in\mathcal{X}_s^n$.
	\item The primary receiver has $M_1M_2N_s$ disjoint decoding sets $\mathcal{D}_{pijl}\subseteq{\mathcal{Y}_p^n}$, $ijl\in\mathcal{M}_1\times\mathcal{M}_2\times\mathcal{N}_s$ and a decoding function $\psi_p: \textbf{y}_p\rightarrow ijl$ if $\textbf{y}_p\in{\mathcal{D}_{pijl}}$, where $ijl\in\mathcal{M}_1\times\mathcal{M}_2\times\mathcal{N}_s$.
	\item The secondary receiver has $L_sN_s$ disjoint decoding sets $\mathcal{D}_{skl}\subseteq{\mathcal{Y}_s^n}$, $kl\in{\mathcal{L}_s\times\mathcal{N}_s}$, and a decoding function $\psi_s: \textbf{y}_p\rightarrow kl$ if $\textbf{y}_p\in\mathcal{D}_{skl}$, where $kl\in\mathcal{L}_s\times\mathcal{N}_s$ (see Fig.\ref{fig:RSModel}).
	\item Probability of error for primary network and secondary link are less than $\epsilon$, that is $Pe_p^o\leq{\epsilon}$ and $Pe_s^o\leq{\epsilon}$ respectively, where
	\begin{multline}
	\label{eq:EprobCRSp}
Pe_p^o=\\	\frac{1}{M_1M_2L_sN_s}\sum_{i,j,k,l}\omega_p^n(\textbf{y}_p\notin\mathcal{D}_{pijl}|\textbf{x}_{1i}\textbf{x}_{2j}\textbf{x}_{skl}),
	\end{multline}
	\begin{multline}
	\label{eq:EprobCRSs}
Pe_s^o=\\	\frac{1}{M_1M_2L_sN_s}\sum_{i,j,k,l}\omega_s^n(\text{y}_s\notin\mathcal{D}_{skl}|\textbf{x}_{1i}\textbf{x}_{2j}\textbf{x}_{skl}).
	\end{multline} 
\end{enumerate}

\begin{figure}
	\centering
		\includegraphics[width=0.45\textwidth]{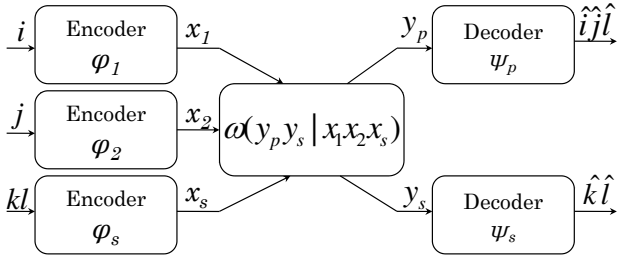}
	\caption{Rate-Splitting channel model $C_{RS}$}
	\label{fig:RSModel}
\end{figure}

A rate tuple $(R_1,R_2,S,T)$ of non-negative real values is achievable over the channel $C_{RS}$ if there exists a code $(n,M_1,M_2,L_s,N_s,\epsilon)$ such that for any arbitrary $0< \epsilon < 1$ and $\eta>0$

\begin{eqnarray}
\label{eq:R^o_F}
\frac{1}{n}\log{M_1}\geq{R_1-\eta},\\
\frac{1}{n}\log{M_2}\geq{R_2-\eta},\\
\frac{1}{n}\log{L_s}\geq{S-\eta},\\
\label{eq:R^o_L}
\frac{1}{n}\log{N_s}\geq{T-\eta},
\end{eqnarray} 
with sufficiently large $n$.

\begin {lemma}
\label{lem:equivalence}
If a rate tuple $(R_1,R_2,S,T)$ is achievable for $C_{RS}$, then a rate tuple $(R_1,R_2,R_s)$ where $R_s=S+T$ is achievable for $C_B$.
\end{lemma}
\begin{proof}

It is sufficient to show that, if $(n,M_1,M_2,L_s,N_s,\epsilon)$ is a code for $C_{RS}$ then $(n,M_1,M_2,L_sN_s,\epsilon)$ is a code for $C_B$. To do so, let $\mathcal{D}_{pij}=\cup_{l=1}^{N_s}\mathcal{D}_{pijl}$. Then
\begin{equation}
\label{eq:Omegap}
\omega_p^n(\textbf{y}_p\notin\mathcal{D}_{pij}|\textbf{x}_{1i}\textbf{x}_{2j}\textbf{x}_{skl})\leq\omega_p^n(\textbf{y}_p\notin\mathcal{D}_{pijl}|\textbf{x}_{1i}\textbf{x}_{2j}\textbf{x}_{skl}).
\end{equation}
So, if $(n,M_1,M_2,L_s,N_s,\epsilon)$ is a code for $C_{RS}$ then $Pe_p^o\leq{\epsilon}$ and $Pe_s^o\leq\epsilon$, hence, from \eqref{eq:Omegap} $Pe_p\leq{\epsilon}$ and $Pe_s\leq\epsilon$ when $k$ and $M_s$ of \eqref{eq:EprobCBp} and \eqref{eq:EprobCBs} are replaced with $kl$ and $L_sN_s$ respectively, meaning that $(n,M_1,M_2,L_sN_s,\epsilon)$ is a code for $C_B$. 
\end{proof}
\subsection{Rate-Splitting Channel with Decodable Primary Signal at the Secondary Receiver}
We introduce another channel, $C_{RS}^p$, in which the secondary user splits its set of messages into two sets, exactly as the case of $C_{RS}$. However, we assume that the signal of one primary transmitter is decodable at the secondary receiver. Without loss of generality, assume this this is the first primary transmitter. Thus, $C_{RS}^p$ is defined by a tuple $(\mathcal{X}_1,\mathcal{X}_2,\mathcal{X}_s,\omega,\mathcal{Y}_p,\mathcal{Y}_s)$ with its elements defined as in $C_B$ and $C_{RS}$. A code $(n,M_1,M_2,L_s,N_s)$ over the channel $C_{RS}^p$ is a collection of $M_1$, $M_2$, $L_sN_s$ codewords such that conditions 1), 2) and 3) of the same code but in $C_{RS}$ are satisfied besides the following two conditions:
\begin{enumerate}
	\item Secondary receiver has $M_1L_sN_s$ disjoint decoding sets $\mathcal{D}_{sikl}\subseteq\mathcal{Y}_s^n$, and a decoding function $\psi_{s}:\textbf{y}_s\rightarrow ikl$ if $\textbf{y}_s\in\mathcal{D}_{sikl}$, where $ikl\in\mathcal{M}_1\times\mathcal{L}_s\times\mathcal{N}_s$.
	\item Probability of error for the primary network and the secondary link are less than $\epsilon$, that is, $Pe_p^r\leq{\epsilon}$ and $Pe_s^r\leq{\epsilon}$ respectively, where
	\begin{multline}
	\label{eq:EprobCRS^pp}
	Pe_p^r=\\	\frac{1}{M_1M_2L_sN_s}\sum_{i,j,k,l}\omega_p^n(\textbf{y}_p\notin\mathcal{D}_{pijl}|\textbf{x}_{1i}\textbf{x}_{2j}\textbf{x}_{skl}),
\end{multline}
\begin{multline}
	\label{eq:EprobCRS^ps}
Pe_s^r=\\
\frac{1}{M_1M_2L_sN_s}\sum_{i,j,k,l}\omega_s^n(\textbf{y}_s\notin\mathcal{D}_{sikl}|\textbf{x}_{1i}\textbf{x}_{2j}\textbf{x}_{skl}).
	\end{multline}
\end{enumerate}
	A rate tuple $(R_1,R_2,S,T)$ of non-negative real values is achievable over the channel $C_{RS}^p$ if for any arbitrary $\eta>0$ and $0<\epsilon<1$ the inequalities \eqref{eq:R^o_F}-\eqref{eq:R^o_L} are satisfied for sufficiently large $n$.	

\begin{lemma}
\label{lem:equivalence^p}
If a rate tuple $(R_1, R_2, S, T)$ is achievable for $C_{RS}^p$, then a rate tuple $(R_1,R_2,R_s)$ where $R_s=S+T$ is achievable for $C_B$.
\end{lemma}
\begin{proof}
The proof follows exactly as the proof of Lemma \ref{lem:equivalence} noting that, if $\mathcal{D}_{skl}=\cup_{i=1}^{M_1}\mathcal{D}_{sikl}$, then
\begin{equation}
\label{eq:Omegas}
\omega_s^n(\textbf{y}_s\notin\mathcal{D}_{skl}|\textbf{x}_{1i}\textbf{x}_{2j}\textbf{x}_{skl})\leq\omega_s^n(\textbf{y}_s\notin\mathcal{D}_{sikl}|\textbf{x}_{1i}\textbf{x}_{2j}\textbf{x}_{skl}).
\end{equation}
\end{proof}

\section{Achievable Rate Region}
\label{sec:ARR}
In this section we consider the characterization of the achievable rate region for $C_B$. In order to do so, we first establish two achievable rate regions, one for $C_{RS}$ and another for $C_{RS}^p$. Then, we define the achievable rate region for $C_B$.
We consider the random variables $U$, $W$ and $Q$ defined over the finite sets $\mathcal{U}$, $\mathcal{W}$ and $\mathcal{Q}$ respectively, where $Q$ is a time sharing parameter. Let the set $\mathcal{P}^*$ contains all $Z=QUWX_1X_2X_sY_pY_s$ such that:

\begin{itemize}
	\item $X_1$, $X_2$, $U$ and $W$ are conditionally independent given $Q$,
	\item $X_s=f(UW|Q)$,
\end{itemize}

Since $X_s=f(UW|Q)$, then $\mathcal{U}$ and $\mathcal{W}$ can be considered as input sets to the channels $C_{RS}$ and $C_{RS}^p$. We establish achievable rate regions for $C_{RS}$ and $C_{RS}^p$ as follows.

\subsection{Achievable Rate Region for $C_{RS}$}
\label{subsec:ARR_CRS}
\begin{theorem}
\label{th:ARR_CRS}
For any $Z\in\mathcal{P}^*$, $\delta^o(Z)$ is the set of achievable rate tuples $(R_1,R_2,S,T)$ for $C_{RS}$ if the following inequalities are satisfied:
\begin{eqnarray}
\label{eq:g_relationsF}
R_1\leq{I(Y_p;X_1|WX_2Q)},\\
R_2\leq{I(Y_p;X_2|WX_1Q)},\\
\label{eq:T^o}
T\leq{I(Y_p;W|X_1X_2Q)},\\
\label{eq:R_pmax}
R_1+R_2\leq{I(Y_p;X_1X_2|WQ)},\\
T+R_1\leq{I(Y_p;WX_1|X_2Q)},\\
T+R_2\leq{I(Y_p;WX_2|X_1Q)},\\
\label{eq:g_relationsMp}
T+R_1+R_2\leq{I(Y_p;WX_1X_2|Q)};\\
\label{eq:g_relationsMs}
S\leq{I(Y_s;U|WQ)},\\
T\leq{I(Y_s;W|UQ)},\\
\label{eq:g_relationsL}
S+T\leq{I(Y_s;UW|Q)}.
\end{eqnarray}
\end{theorem}
\begin{proof}
Please refer to Appendix \ref{App:ARR_CRRS}
\end{proof}

\begin{corollary}
For $\delta^o=\cup_{Z\in\mathcal{P}^*}\delta^o(Z)$, any rate tuple of $\delta^o$ is achievable.
\end{corollary}

In the defined network we focus on the achievable rates by the primary network $R_p=R_1+R_2$ and the secondary link $R_s=S+T$. Let $\mathcal{R}^o(Z)$ be the set of all rate tuples $(R_s,R_p)$ having $(R_1,R_2,S,T)$ satisfy \eqref{eq:g_relationsF}-\eqref{eq:g_relationsL} for all $Z\in\mathcal{P}^*$, then we determine $\mathcal{R}^o(Z)$ in the following theorem.

\begin{theorem}
\label{th:R^o(Z)}
For any $Z\in{\mathcal{P}^*}$ the achievable rate region $\mathcal{R}^o(Z)$ of the defined channel $C_{RS}$ consists of all rate pairs $(R_s,R_p)$ that satisfy
\begin{equation}
\label{eq:rate_regionCRS}
R_p\leq{\rho_p^o}, \quad R_s\leq{\rho_s^o}, \quad R_s+R_p\leq{\rho_{sp}^o} 
\end{equation}
where
\begin{eqnarray}
\label{eq:Rhos^o}
\rho_p^o=I(Y_p;X_1X_2|WQ),\\
\rho_s^o=I(Y_s;U|WQ)+\sigma^*,\\
\label{eq:sumrate}
\begin{split}
\rho_{sp}^o =& \rho_p^o+I(Y_s;U|WQ)\\ 
             &+\min\{I(Y_s;W|Q),I(Y_p,W|Q)\}
\end{split}
\end{eqnarray}
and
\begin{equation}
\label{eq:sigmastar}
\sigma^*=\min\{I(Y_p;W|X_1X_2Q),I(Y_s;W|Q)\}.
\end{equation}
\end{theorem}

\begin{figure}[ht]
	\centering
		\includegraphics[width=0.3\textwidth]{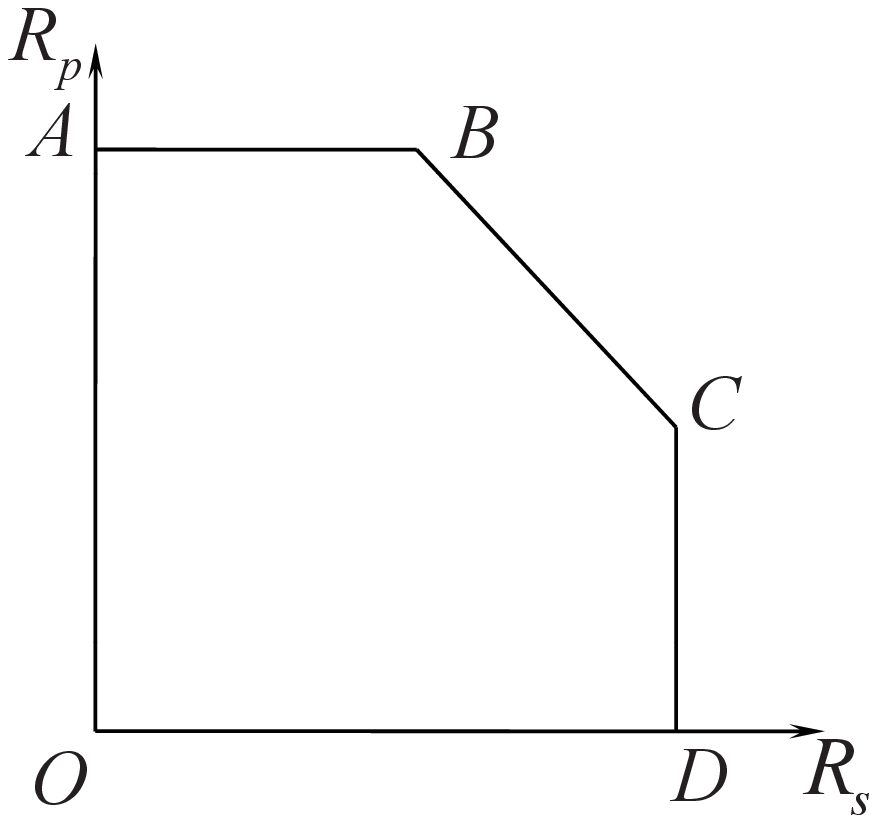}
	\caption{Ahievable rate region $\mathcal{R}^o(Z)$ of the channel $C_{RS}$ for ony $Z\in\mathcal{P}^*$.}
	\label{fig:RR_C_RS}
\end{figure}

\begin{proof}
To proof the theorem it is sufficient to determine the rate tuples $(R_s,R_p)$ of the corner points of $\mathcal{R}^o(Z)$. To do so, we refer to Fig. \ref{fig:RR_C_RS}.
\begin{itemize}
	\item \emph{Point A:}
\end{itemize}
$R_s^{A}=0$, i.e., $S^{A}=T^{A}=0$. Thus the maximum rate at which the primary network can operate is determined from \eqref{eq:R_pmax} as:
\begin{equation}
\label{eq:R_pA^o}
R_p^{A}={I(Y_p;X_1X_2|WQ)}=\rho_p^o
\end{equation}
\begin{itemize}
	\item \emph{Point B:}
\end{itemize}
At this point we find the maximum possible rate at which the secondary user can transmit when the primary rate is $R_p^{B}=\rho_p^o$. In this case the relations of \eqref{eq:g_relationsF}-\eqref{eq:g_relationsL} are reduced to
\begin{eqnarray}
\label{eq:B1}
T\leq{I(Y_p;W|Q)},\\
\label{eq:B2}
\rho_p^o+T\leq{I(Y_p;WX_1X_2|Q)};\\
\label{eq:B3}
T\leq{I(Y_s;W|UQ)},\\
\label{eq:SinB}
S\leq{I(Y_s;U|WQ)},\\
\label{eq:B5}
S+T\leq{I(Y_s;UW|Q)}.
\end{eqnarray}
Since $T$ is irrelevant in \eqref{eq:SinB}, then $S$ can be set to
\begin{equation}
\label{eq:S^Bo}
S^{B}=I(Y_s;U|WQ).
\end{equation}
Hence, using chain rule in \eqref{eq:B2} and \eqref{eq:B5}, the maximum value for $T$ would be
\begin{equation}
\label{eq:T^Bo}
T^{B}=\min\{I(Y_p;W|Q),I(Y_s;W|Q)\}
\end{equation}
and $R_s^{B}=S^{B}+T^{B}$.
\begin{itemize}
	\item \emph{Point D:}
\end{itemize}
$R_1^D=R_2^D=R_p^{D}=0$, then \eqref{eq:g_relationsF}-\eqref{eq:g_relationsL} are reduced to
\begin{eqnarray}
T\leq{I(Y_p;W|X_1X_2Q)};\\
\label{eq:SinD}
S\leq{I(Y_s;U|WQ)},\\
T\leq{I(Y_s;W|UQ)},\\
S+T\leq{I(Y_s;UW|Q)}.
\end{eqnarray}
Since $T$ is irrelevant in \eqref{eq:SinD}, $S$ can be set to
\begin{equation}
\label{eq:S^Do}
S^{D}=I(Y_s;U|WQ).
\end{equation}
Then,
\begin{equation}
\label{eq:T^Do}
T^{D}=\sigma^*=\min\{I(Y_s;W|Q),I(Y_p;W|X_1X_2Q)\}
\end{equation}
and $R_s^{D}=S^{D}+T^{D}=\rho_s^o$.
\begin{itemize}
	\item \emph{Point C:}
\end{itemize}
At $R_s^{C}=\rho_s^o$, the maximum possible primary rate $R_p=R_1+R_2$ has to satisfy
\begin{eqnarray}
\label{eq:C1}
R_p\leq{I(Y_p;X_1X_2|WQ)},\\
\label{eq:C2}
R_p\leq{I(Y_p;WX_1X_2|Q)-\sigma^*}.
\end{eqnarray}
Using chain rule, \eqref{eq:C2} can be rewritten as
\begin{equation}
\label{eq:C2dash}
R_p\leq{I(Y_p;X_1X_2|WQ)+I(Y_p;W|Q)-\sigma^*}.
\end{equation}
Thus, if $I(Y_p;W|Q)-\sigma^*>0$ then \eqref{eq:C2dash} will be dominated by \eqref{eq:C1}. Otherwise, \eqref{eq:C2dash} dominates \eqref{eq:C1}. So, $R_p^{C}$ will be given by,
\begin{equation}
\label{eq:R_pC}
R_p^{C}=I(Y_p;X_1X_2|WQ)-\left[\sigma^*-I(Y_p;W|Q)\right]^+
\end{equation}
where $[x]^+=\max\{0,x\}$.
The following is to show that both points $(R_s^{B},R_p^{B})$ and $(R_s^{C},R_p^{C})$ lie on the line $R_s+R_p={\rho_{sp}^o}$:

	For Point B, using direct substitution with $$R_s^{B}=I(Y_s;U|WQ)+\min\{I(Y_p;W|Q),I(Y_s;W|Q)\}$$ and $$R_p^{B}=\rho_p^o$$ it is clear that $R_s^{B}+R_p^{B}=\rho_{sp}^o$.
	
	For Point C, we consider the following two possibilities:
	
\begin{itemize}
	\item $\sigma^* \geq{I(Y_p;W|Q)}$:
\end{itemize}
	Here $\min\{I(Y_s;W|Q),I(Y_p,W|Q)\}=I(Y_p;W|Q)$. Consequently,
	\begin{equation*}
	\rho_{sp}^o=I(Y_s;U|WQ)+I(Y_p;WX_1X_2|Q)
	\end{equation*}
	and
	\begin{equation*}
	R_s^{C}+R_p^{C}=I(Y_s;U|WQ)+I(Y_p;WX_1X_2|Q).
	\end{equation*}
\begin{itemize}
	\item $\sigma^* <I(Y_p;W|Q)$:
\end{itemize}	
	Since $I(Y_p;W|X_1X_2Q)\geq{I(Y_p;W|Q)}$, therefore $$I(Y_s;W|Q)<I(Y_p;W|Q).$$ Consequently,
	\begin{equation*}
	\rho_{sp}^o=I(Y_s;UW|Q)+I(Y_p;X_1X_2|WQ)
	\end{equation*}
	and
	\begin{equation*}
	R_s^{C}+R_p^{C}=I(Y_s;UW|Q)+I(Y_p;X_1X_2|WQ).
	\end{equation*}
	Therefore, both rate tuples $(R_s^{B},R_p^{B})$ and $(R_s^{C},R_p^{C})$ lie on the line $R_s+R_p=\rho_{sp}^o$.
\end{proof}

Note that, in the Appendix of \cite{HanKob} Han and Kobayashi argued that part of the achievable rate region by their introduced scheme was bounded by lines of slopes $-0.5$ and $-2$. Although from \eqref{eq:g_relationsF}-\eqref{eq:g_relationsL} reducing $T$ by a value of $r$ may result in increase of $R_p$ by $2r$, the proof that point $(R_s^{C},R_p^{C})$ lie on the line $R_s+R_p=\rho_{sp}^o$ means that a bound of slope $-2$ does not exist for $\mathcal{R}^o(Z)$.  

\begin{corollary}
Any rate tuple $(R_s,R_p)$ of the region
\begin{equation}
\label{eq:R^o}
\mathcal{R}^o=\mbox{closure of} \bigcup_{Z\in\mathcal{P}^*}\mathcal{R}^o(Z)
\end{equation}
is achievable.
\end{corollary}

\subsection{Achievable Rate Region for $C_{RS}^p$}
\label{subsec:ARR_CRSp}
Since in $C_{RS}^p$ the signal of one primary user has to be decodable at the secondary receiver, the model of $C_{RS}^p$ can be considered as the modified interference channel model, $C_m$, introduced in \cite{HanKob}. The signals of the two primary users can be treated as if they are produced from single source splitting its signal into two parts and encoding each part separately such that, one part is decodable at both receivers while the other is decodable only at the primary receiver. For this channel, we define the set $\delta_i^r(Z)$ as the set of all achievable rate tuples $(R_1,R_2,S,T)$ when the signal of primary transmitter $i$, $i\in\{1,2\}$, is decodable by the secondary receiver. Without loss of generality, we assume that $i=1$. Then, we define an achievable rate region for $C_{RS}^p$ in the following theorem.
\begin{theorem}
\label{th:ARR_CRSp}
For any $Z\in\mathcal{P}^*$, $\delta_1^r(Z)$ is the set of achievable rate tuples $(R_1,R_2,S,T)$ over the channel $C_{RS}^p$ if the following inequalities are satisfied:
\begin{eqnarray}
\label{eq:gr_relationsF}
R_1\leq I(Y_p;X_1|WX_2Q),\\
R_2\leq I(Y_p;X_2|WX_1Q),\\
\label{eq:T^r}
T\leq I(Y_p;W|X_1X_2Q),\\
R_1+R_2\leq I(Y_p;X_1X_2|WQ),\\
R_1+T\leq I(Y_p;WX_1|X_2Q),\\
R_2+T\leq I(Y_p;WX_2|X_1Q),\\
\label{eq:gr_relationsMp}
R_1+R_2+T\leq I(Y_p;WX_1X_2Q);\\
\label{eq:gr_relationsMs}
S\leq I(Y_s;U|WX_1Q),\\
T\leq I(Y_s;W|UX_1Q),\\
\label{eq:R_1^r}
R_1\leq I(Y_s;X_1|UWQ),\\
S+T\leq I(Y_s;UW|X_1Q),\\
R_1+S\leq I(Y_s;UX_1|WQ),\\
R_1+T\leq I(Y_s;WX_1|UQ),\\
\label{eq:gr_relationsL}
R_1+S+T\leq I(Y_s;UWX_1|Q).
\end{eqnarray}
\end{theorem}
\begin{proof}
The proof follows exactly as the proof of Theorem 3.1 in \cite{HanKob}.
\end{proof}

\begin{corollary}
For $\delta_1^r=\cup_{Z\in\mathcal{P}^*}\delta_1^r(Z)$, any rate tuple of $\delta_1^r$ is achievable.
\end{corollary}

For $C_{RS}^p$ we define the region $\mathcal{R}_i^r(Z)$ as the set of rate tuples $(R_s,R_p)$ where $R_s=S+T$, $R_p=R_1+R_2$ and $(R_1,R_2,S,T)$ is an element of $\delta_i^r(Z)$ for any $Z\in\mathcal{P}^*$, $i\in\{1,2\}$.
\begin{theorem}
\label{th:R^r(Z)}
For any $Z\in\mathcal{P}^*$ the achievable rate region $\mathcal{R}_1^r(Z)$ for the channel $C_{RS}^p$ consists of all rate pairs $(R_s,R_p)$ that satisfy
\begin{equation}
\begin{aligned}
R_s\leq\rho_s^r, \ \ \ \ &  R_p\leq\rho_p^r, & R_s+R_p\leq\rho_{sp}^r,\\
2R_s+R_p\leq\rho_{2p}^r, &  & R_s+2R_p\leq\rho_{s2}^r
\end{aligned}
\end{equation}
where
\begin{eqnarray}
\label{eq:rhos^r}
\rho_s^r=I(Y_s;U|WX_1Q)+\sigma_s^*,\\
\rho_p^r=I(Y_p;X_2|WX_1Q)+\sigma_p^*,
\end{eqnarray}
\begin{equation}
\begin{split}
\rho_{sp}^r=& I(Y_s;U|WX_1Q)+I(Y_p;X_2|WX_1Q)+\\
						&+\min\{I(Y_p;WX_1|Q),I(Y_s;WX_1|Q),\\
						& I(Y_p;W|X_1Q)+I(Y_s;X_1|WQ),\\
						& I(Y_p;X_1|WQ)+I(Y_s;W|X_1Q)\},
\end{split}
\end{equation}
\begin{equation}
\begin{split}
\rho_{2p}^r=& 2 I(Y_s;U|WX_1Q)+2\sigma_s^*+I(Y_p;X_2|WX_1Q)\\
            &-\left[\sigma_s^*-I(Y_p;W|X_1Q)\right]^++\min\{I(Y_s;X_1|WQ),\\
            &I(Y_s;WX_1|Q)-\sigma_s^*, I(Y_p;X_1|Q)\\
            &+\left[I(Y_p;W|X_1Q)-\sigma_s^*\right]^+,I(Y_p;X_1|WQ)\},
\end{split}
\end{equation}
\begin{equation}
\begin{split}
\rho_{s2}^r=& 2 I(Y_p;X_2|WX_1Q)+2\sigma_p^*+I(Y_s;U|WX_1Q)\\
            &-\left[\sigma_p^*-I(Y_s;X_1|WQ)\right]^++\min\{I(Y_p;W|X_1Q),\\
            &I(Y_p;WX_1|Q)-\sigma_p^*, I(Y_s;W|Q)\\
            &+\left[I(Y_s;X_1|WQ)-\sigma_p^*\right]^+,I(Y_s;W|X_1Q)\},
\end{split}
\end{equation}
and
\begin{eqnarray}
\sigma_s^*=\min\{I(Y_s;W|X_1Q),I(Y_p;W|X_1X_2Q)\},\\
\sigma_p^*=\min\{I(Y_p;X_1|WQ),I(Y_s;X_1|UWQ)\}
\end{eqnarray}
as shown in Fig. \ref{fig:RR_CRSp}
\end{theorem}

\begin{figure}[ht]
	\centering
		\includegraphics[width=0.35\textwidth]{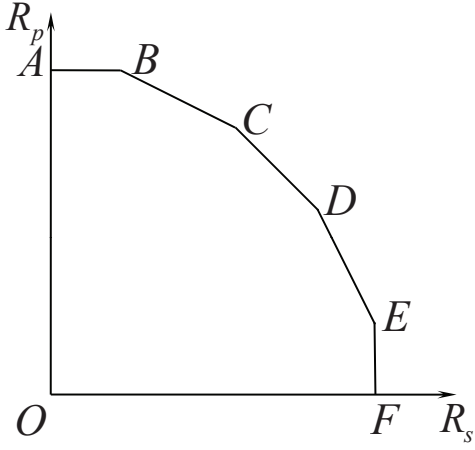}
	\caption{Ahievable rate region $\mathcal{R}_1^r(Z)$ of the channel $C_{RS}^p$ for $Z\in\mathcal{P}^*$.}
	\label{fig:RR_CRSp}
\end{figure}

\begin{proof}
From the similarity between $C_{RS}^p$ and the modified interference channel of Han and Kobayashi \cite{HanKob}, the derivation of the achievable rate region can be found in the Appendix of \cite{HanKob}. The analysis basically goes as that done for $\mathcal{R}^o(Z)$ in $C_{RS}$. In this proof we directly mention the corner points of the $\mathcal{R}_1^r(Z)$ shown in Fig. \ref{fig:RR_CRSp} as follows.
\begin{itemize}
	\item Point A: 
\end{itemize}
\begin{equation}
\label{eq:Rs^rA}
R_s^{A}=0,
\end{equation}
\begin{equation}
\label{eq:Rp^rA}
R_p^{A}=\rho_p^r=I(Y_p;X_2|X_1WQ)+\sigma_p^*.
\end{equation}
\begin{itemize}
	\item Point B:
\end{itemize}
\begin{equation}
\label{eq:Rs^rB}
\begin{split}
R_s^{B}=& I(Y_s;U|WX_1Q)-[\sigma_p^*-I(Y_s;X_1|WQ)]^+\\
				 &+\min\{I(Y_p;W|X_1Q),I(Y_p;WX_1|Q)-\sigma_p^*,\\
				 &I(Y_s;W|Q)+\left[I(Y_s;X_1|WQ)-\sigma_p^*\right]^+,\\
				 &I(Y_s;W|X_1Q)\},
\end{split}
\end{equation}
\begin{equation}
\label{eq:Rp^rB}
R_p^{B}=\rho_p^r=I(Y_p;X_2|X_1WQ)+\sigma_p^*.
\end{equation}
\begin{itemize}
	\item Point C:
\end{itemize}
\begin{equation}
\label{eq:Rs^rC}
R_s^{C}=2\rho_{sp}^r-\rho_{s2}^r,\\
\end{equation}
\begin{equation}
\label{eq:Rp^{rC}}
R_p^{C}=\rho_{s2}^r-\rho_{sp}^r.
\end{equation}
\begin{itemize}
	\item Point D:
\end{itemize}
\begin{equation}
\label{eq:Rs^rD}
R_s^{D}=\rho_{2p}^r-\rho_{sp}^r,\\
\end{equation}
\begin{equation}
\label{eq:Rp^rD}
R_p^{D}=2\rho_{sp}^r-\rho_{sp}^r.
\end{equation}
\begin{itemize}
	\item Point E:
\end{itemize}
\begin{equation}
\label{eq:Rs^rE}
R_s^{E}=I(Y_s;U|WX_1Q)+\sigma_s^*,
\end{equation}
\begin{equation}
\label{eq:Rp^rE}
\begin{split}
R_p^{E}=&I(Y_p;X_2|WX_1Q)-[\sigma_s^*-I(Y_p;W|X_1Q)]^+\\
         &+\min\{I(Y_s;X_1|WQ),I(Y_s;WX_1|Q)-\sigma_s^*,\\
          &I(Y_p;X_1|Q)+\left[I(Y_p;W|X_1Q)-\sigma_s^*\right]^+,\\
          &I(Y_p;X_1|WQ)\}.
\end{split}
\end{equation}
\begin{itemize}
	\item Point F:
\end{itemize}
\begin{equation}
\label{eq:Rs^F}
R_s^{rF}=\rho_s^r=I(Y_s;U|WX_1Q)+\sigma_s^*,
\end{equation}
\begin{equation}
\label{eq:Rp^rF}
R_p^{F}=0.
\end{equation}
\end{proof}
\begin{corollary}
Any rate tuple $(R_s,R_p)$ of the region
\begin{equation}
\label{eq:R^r}
\mathcal{R}_1^r=\mbox{closure} \bigcup_{Z\in\mathcal{P}^*} \mathcal{R}_1^r(Z)
\end{equation}
is achievable.
\end{corollary}

Constraining the signal of one primary user to be decodable at the secondary receiver might result in a degradation in the achievable primary rate especially when the secondary rate is very small. In general $\mathcal{R}^o$ and $\mathcal{R}_i^r$ do not necessarily include one another. However, there exists a case for which $\mathcal{R}^o\subseteq\mathcal{R}_i^r$. To characterize this case we introduce the following theorem. 
\begin{theorem}
\label{th:inclusion}
 For a given $Z\in\mathcal{P}^*$, $\mathcal{R}^o(Z)\subseteq\mathcal{R}_i^r(Z)$ if and only if
 \begin{equation}
 \label{eq:cond}
 I(Y_p;X_i|WQ)\leq I(Y_s;X_i|UWQ).
 \end{equation} 
\end{theorem}

\begin{proof}
Please refer to Appendix \ref{app:RRproof}.
\end{proof}

\begin{corollary}
\label{cor:Inclusion}
If for all $Z\in\mathcal{P}^*$ condition \eqref{eq:cond} is satisfied, then $\mathcal{R}^o\subseteq\mathcal{R}_i^r$, where $\mathcal{R}_i^r=\cup_{Z\in\mathcal{P}^*}\mathcal{R}_i^r(Z)$.
\end{corollary}

Theorem \ref{th:inclusion} shows that when a primary user encodes its messages at a rate decodable at both receivers, the primary network may achieve the same rate range when none of the signal of its users is decodable at the secondary receiver. Moreover, at every primary rate the secondary rate is enhanced (see Fig.\ref{fig:RRproof}). Hence, we conclude the following Proposition.
\begin{proposition}
If for any $Z\in\mathcal{P}^*$ condition \eqref{eq:cond} is satisfied, then allowing the secondary receiver to decode the signal of primary user $i$ at this $Z$ enhances the range of the secondary achievable rates without affecting the range of the achievable primary sum rates.
\end{proposition}

We call Corollary \ref{cor:Inclusion} \emph{Primary Decodability Condition (PDC)}.

\subsection{Achievable Rate Region for the Channel $C_B$}

From $C_{RS}$ and $C_{RS}^p$ we define
\begin{equation}
\label{eq:R_i(Z)}
\mathcal{R}_i(Z)=\mathcal{R}^o(Z)\cup \mathcal{R}_i^r(Z), \quad Z\in\mathcal{P}^*, i\in\{1,2\},
\end{equation}
and
\begin{equation}
\label{eq:R_i}
\mathcal{R}_i=\text{closure} \bigcup_{Z\in\mathcal{P}^*}\mathcal{R}_i(Z),\quad i\in\{1,2\}.
\end{equation}
Hence, an achievable rate region for the channel $C_B$
\begin{equation}
\label{eq:R}
\mathcal{R}=\mathcal{R}_1\cup\mathcal{R}_2,
\end{equation}
or equivalently,
\begin{equation}
\label{eq:Req}
\mathcal{R}=\mathcal{R}^o\cup\mathcal{R}_1^r\cup\mathcal{R}_2^r.
\end{equation}

Note that, inequalities \eqref{eq:T^o} and \eqref{eq:T^r} used in $\delta^o(Z)$ and $\delta_1^r(Z)$, assuming $i=1$, respectively, to limit the error in decoding the public part of the secondary signal at the primary receiver while the primary signals are decoded successfully. In fact, the primary receiver may not be interested in limiting the probability of such error event. Similarly, inequality \eqref{eq:R_1^r} in $\delta_1^r(Z)$ may not be relevant as the secondary receiver is not interested in limiting the probability of error in decoding the primary signal when the two parts of its signal are decoded successfully. However, removing \eqref{eq:T^o} from the definition of $\delta^o(Z)$ and \eqref{eq:T^r} and \eqref{eq:R_1^r} from the definition of $\delta_1^r(Z)$ does not enhance the achievable rate region $\mathcal{R}$.

To demonstrate this fact, we define $\delta'(Z)$ exactly as $\delta(Z)$ but without the constraint of \eqref{eq:T^o}, and $\delta_1'^r(Z)$ exactly as $\delta_1^r(Z)$ but without the constraints \eqref{eq:T^r} and \eqref{eq:R_1^r}. Let $\mathcal{R}'^o(Z)$ and $\mathcal{R}_1'^r(Z)$ be two sets of rate tuples $(R_s,R_p)$ such that $R_s=S+T$ and $R_p=R_1+R_2$ and the rate tuple $(R_1,R_2,S,T)$ is an element of $\delta'^o(Z)$ and $\delta_1'^r(Z)$, respectively. Also we define $$\mathcal{R}_1'(Z)=\mathcal{R}'^o(Z)\cup\mathcal{R}_1'^r(Z).$$

\begin{theorem}
\label{th:noextention}
If $\mathcal{R}_1'=\bigcup_{Z\in\mathcal{P}^*}\mathcal{R}_1'(Z)$, then $\mathcal{R}_1'=\mathcal{R}_1$.
\end{theorem}
\begin{proof}
Please refer to Appendix \ref{App:nxt}.
\end{proof}

\begin{corollary}
\label{cor:noex5tention}
For $$\mathcal{R}'=\text{closure of }\mathcal{R}'_1\cup\mathcal{R}'_2,$$ then $$\mathcal{R}'=\mathcal{R}.$$
\end{corollary}

\section{Gaussian Channel}
\label{sec:Gauss}

In this section we quantify the obtained achievable rate regions in a Gaussian channel model. A memoryless Gaussian channel of the introduced system is defined by a tuple $(\mathcal{X}_1,\mathcal{X}_2,\mathcal{X}_s,\omega,\mathcal{Y}_p,\mathcal{Y}_s)$ with $\mathcal{X}_1=\mathcal{X}_2=\mathcal{X}_s=\mathcal{Y}_p=\mathcal{Y}_s=\Re$ (the field of real numbers), and a channel probability $\omega$ specified by,
\begin{eqnarray}
y_p=\sqrt{g_1^p}x_1+\sqrt{g_2^p}x_2+\sqrt{g_s^p}x_s+n_p,\\
y_s=\sqrt{g_1^s}x_1+\sqrt{g_2^s}x_2+\sqrt{g_s^s}x_s+n_s
\end{eqnarray}
for $x_1\in\mathcal{X}_1$, $x_2\in\mathcal{X}_2$, $x_s\in\mathcal{X}_s$, $y_p\in{\mathcal{Y}_p}$ and $y_s\in\mathcal{Y}_s$, where $n_p$ and $n_s$ are independent Gaussian additive noise samples with zero mean and variance $N_0$, and $g_1^p$, $g_2^p$, $g_s^p$, $g_1^s$, $g_2^s$ and $g_s^s$ are the channel power gains. Power constraints are imposed on codewords $\textbf{x}_1(i)$, $\textbf{x}_2(j)$, $\textbf{x}_s(k)$ ($i\in{\mathcal{M}_1}$, $j\in{\mathcal{M}_2}$, $k\in\mathcal{M}_s$):
\begin{eqnarray}
\frac{1}{n}\sum_{t=1}^{n}{(x_1(i)^{(t)})^2}={P_1},\\
\frac{1}{n}\sum_{t=1}^{n}{(x_2(j)^{(t)})^2}={P_2},\\
\frac{1}{n}\sum_{t=1}^{n}{(x_s(k)^{(t)})^2}={P_s}.
\end{eqnarray}

For computation, we define a subclass $\mathcal{G}(P_1,P_2,P_s)$ of $\mathcal{P}^*$ as follows: $Z=\phi UWX_1X_2X_sY_pY_s\in{\mathcal{G}(P_1,P_2,P_s)}$ if and only if $Z\in{\mathcal{P}^*}$, $\sigma^2(X_1)={P_1}$, $\sigma^2(X_2)={P_2}$ and $\sigma^2(X_s)={P_s}$ with $X_1$, $X_2$, $U$ and $W$ are zero mean Gaussian and $X_s=U+W$.
Hence, we have the following rate regions achievable:
\begin{eqnarray}
\mathcal{R}^o_g&=&\text{closure of } \bigcup_{Z\in\mathcal{G}(P_1,P_2,P_s)} \mathcal{R}^o(Z),\\
\mathcal{R}^r_{ig}&=&\text{closure of } \bigcup_{Z\in\mathcal{G}(P_1,P_2,P_s)} \mathcal{R}^r_i(Z), i\in\{1,2\},\\
\mathcal{R}_{ig}&=&\text{closure of }\bigcup_{Z\in\mathcal{G}(P_1,P_2,P_s)} \mathcal{R}_{i}(Z), i\in\{1,2\},\\
\mathcal{R}_g&=&\mathcal{R}^o_g\bigcup \left(\cup_{i\in\{1,2\}} \mathcal{R}^r_{ig}\right)= \mathcal{R}_{1g}\bigcup\mathcal{R}_{2g}.
\end{eqnarray}

Assume the secondary user splits its power into $\lambda P_s$ and $\bar{\lambda} P_s$ such that $0\leq\lambda\leq 1$ and $\lambda+\bar{\lambda}=1$. The part of secondary signal decodable at the primary and secondary receivers is encoded with power $\bar{\lambda} P_s$ where the other part is encoded with power $\lambda P_s$. Let $\tau(x)=0.5\log_2(1+x)$, the relevant quantities in Theorems \ref{th:R^o(Z)} and \ref{th:R^r(Z)} will be given by:

\begin{align*}
&I(Y_p;X_1X_2|W)=\tau \left(\frac{g_1^pP_1+g_2^pP_2}{g_s^p\lambda P_s+N_0}\right),\\
&I(Y_p;X_1X_2)=\tau \left(\frac{g_1^pP_1+g_2^pP_2}{g_s^p P_s+N_0}\right),\\
&I(Y_p;X_2|WX_1)=\tau \left(\frac{g_2^pP_2}{g_s^p\lambda P_s+N_0}\right),\\
&I(Y_p;X_1|W)=\tau \left(\frac{g_1^pP_1}{g_s^p\lambda P_s+g_2^pP_2+N_0}\right),\\
&I(Y_p;W|X_1X_2)=\tau \left(\frac{g_s^p\bar{\lambda}P_s}{g_s^p\lambda P_s+N_0}\right),\\
&I(Y_p;W|X_1)=\tau \left(\frac{g_s^p\bar{\lambda}P_s}{g_s^p\lambda P_s+g_2^p P_2+N_0}\right),\\
&I(Y_p;WX_1)=\tau \left(\frac{g_1^p P_1+g_s^p\bar{\lambda}P_s}{g_s^p\lambda P_s+g_2^p P_2+N_0}\right),\\
&I(Y_p;W)=\tau \left(\frac{g_s^p\bar{\lambda}P_s}{g_s^p\lambda P_s+g_1^pP_1+g_2^pP_2+N_0}\right)\\
&I(Y_p;X_1)=\tau \left(\frac{g_1^p P_1}{g_s^p P_s+g_2^p P_2+N_0}\right);\\
&I(Y_s;U|WX_1)=\tau\left(\frac{g_s^s\lambda P_s}{g_2^sP_2+N_0}\right),\\
&I(Y_s;U|W)=\tau\left(\frac{g_s^s\lambda P_s}{g_1^sP_1+g_2^sP_2+N_0}\right),\\
&I(Y_s;W|X_1)=\tau\left(\frac{g_s^s\bar{\lambda} P_s}{g_s^s\lambda P_s+g_2^sP_2+N_0}\right),\\
&I(Y_s;WX_1)=\tau\left(\frac{g_s^s\bar{\lambda}P_s+g_1^sP_1}{g_s^s\lambda P_s+g_2^sP_2+N_0}\right),\\
&I(Y_s;W)=\tau\left(\frac{g_s^s\bar{\lambda} P_s}{g_s^s\lambda P_s+g_1^sP_1+g_2^sP_2+N_0}\right),\\
&I(Y_s;X_1|W)=\tau\left(\frac{g_1^sP_1}{g_s^s\lambda P_s+g_2^sP_2+N_0}\right),\\
&I(Y_s;X_1|UW)=\tau\left(\frac{g_1^sP_1}{g_2^sP_2+N_0}\right).
\end{align*}

\subsection{Performance of Rate-Splitting}

In this subsection we study the effect of rate-splitting by the secondary link on the achievable rate regions $\mathcal{R}^o_g$ and $\mathcal{R}^r_{ig}$, $i\in\{1,2\}$ and hence $\mathcal{R}_{ig}$. For each region there exists a case for which no rate-splitting determines the overall region, i.e., each achievable rate region is obtained at $\lambda=0$ or $\lambda=1$. We say that rate-splitting does not affect an achievable rate region $\mathcal{A}$ if $\mathcal{A}(Z)$ coincides on $\mathcal{A}$ at $\lambda=0$ or $\lambda=1$, $Z\in\mathcal{G}(P_1,P_2,P_s)$, where $\mathcal{A}=\bigcup_{Z\in\mathcal{G}(P_1,P_2,P_s)}\mathcal{A}(Z)$, meaning that either decoding the whole secondary signal at the primary receiver or not decoding it at all determines $\mathcal{A}$.

\subsubsection{For $\mathcal{R}_g^o$}

The region $\mathcal{R}^o_g$ is obtained when the secondary receiver is assumed to treat the primary interference as noise. The following theorem determines the effect of rate-splitting on $\mathcal{R}^o_g$. 
\begin{theorem}
\label{th:R^o_g}
For $Z\in\mathcal{G}(P_1,P_2,P_s)$, an achievable rate region $\mathcal{R}^o(Z)$ coincides on $\mathcal{R}^o_g$ if and only if $\lambda=0$ and
\begin{equation}
I(Y_s;W)\leq{I(Y_p;W|X_1X_2)}
\end{equation}
or equivalently,
\begin{equation}
\label{eq:cond2}
g_s^s N_0\leq g_s^p(g_1^sP_1+g_2^sP_2+N_0).
\end{equation}
\end{theorem}
\begin{proof}
Please refer to Appendix \ref{App:R^o_g}.
\end{proof}

Theorem \ref{th:R^o_g} shows that rate-splitting does not affect the achievable rate region $\mathcal{R}^o_g$ when inequality \eqref{eq:cond2} is satisfied. Hence, a primary receiver decoding all the secondary signal is preferable at this case. Fig. \ref{fig:R^o_g_C1} depicts this case for different values of $\lambda$. It is clear that $\mathcal{R}^o(Z)$ at smaller $\lambda$ contains $\mathcal{R}^o(Z)$ at larger $\lambda$. This figure was obtained at $g_1^p=2.5664$, $g_2^p=3.7653$, $g_1^s=0.1812$, $g_2^s=0.1784$, $g_s^p=2.3620$ and $g_s^s=8.6065$, and at the following power setup. The noise variance $N_0=1$ unit power and $\frac{P_1}{N_0}=\frac{P_2}{N_0}=\snr_p=10$ dB and $\frac{P_s}{N_0}=\snr_s=10$ dB. Note that, in this case the maximum secondary throughput does not depend on $\lambda$, so the best performance from the primary rate point of view is to decode all the secondary signal by setting $\lambda=0$.

\begin{figure}[ht]
  \centering
  \subfloat[The overall achievable rate region $\mathcal{R}^o_g$ is obtained when the whole secondary signal is decodable by the primary receiver. $\mathcal{R}^o_g$ is shown in blue.]{\label{fig:R^o_g_C1}\includegraphics[width=0.3\textwidth]{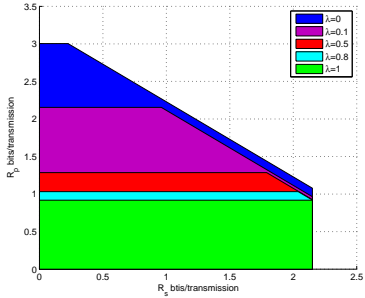}}\\                
  \subfloat[Rate-splitting affects the achievable rate region. $\mathcal{R}^o_g$ is shown in blue and $\mathcal{R}^o(Z)$ is shown in green for $\lambda=0$, yellow for $\lambda=0.1$ and red  for $\lambda=1$.]{\label{fig:R^o_g_C2}\includegraphics[width=0.3\textwidth]{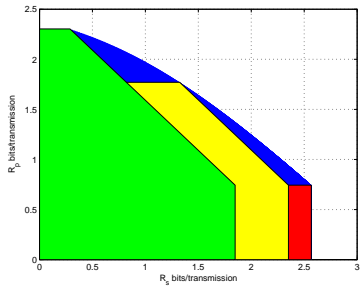}}
  
  \caption{Performance of rate-splitting by the secondary link when it treats the primary interference as noise.}
  \label{fig:RS_performance1}
\end{figure}

Moreover, when inequality \eqref{eq:cond2} is not satisfied, rate-splitting affects $\mathcal{R}^o_g$ as for any two different values of $\lambda$ the corresponding $\mathcal{R}^o(Z)$s do not contain one another. Hence, $\mathcal{R}^o_g$ is obtained by varying $\lambda$ from $0$ to $1$. Fig. \ref{fig:R^o_g_C2} represents the case when \eqref{eq:cond2} is not satisfied for the following parameters. $g_1^p=1.5066$, $g_2^p=0.8290$, $g_1^s=0.1902$, $g_2^s=0.0122$, $g_s^p=1.1953$ and $g_s^s=10.3229$ with the same power setup of Fig. \ref{fig:R^o_g_C1}.

Also, it is shown in \cite{ICT} that when \eqref{eq:cond2} is not satisfied, then the sum throughput of the whole network, i.e., $R_s+R_p$ increases with $\lambda$. That is, as $\lambda$ increases the primary sum rate decreases but the secondary rate gains an increase larger than the decrease in rate encountered by the primary network. Fig. \ref{fig:JournalRateIncrease} depicts $R_s+R_p$ for the same simulation parameters of Fig. \ref{fig:R^o_g_C2}. It is clear that the increase in the total sum rate, $R_s+R_p$, is accompanied by a decrease in the sum primary rate $R_p$. Hence, the sum primary rate has to be protected above a minimum limit.
\begin{figure}
	\centering
		\includegraphics[width=0.40\textwidth]{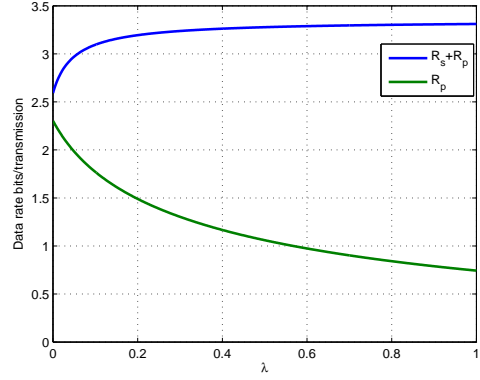}
	\caption{Increase in the sum rate of the whole network when  inequality \eqref{eq:cond2} is not satisfied.}
	\label{fig:JournalRateIncrease}
\end{figure}

\subsubsection{For $\mathcal{R}^r_{ig}$, $i\in\{1,2\}$}
The region $\mathcal{R}^r_{ig}$ is obtained when the secondary receiver can decode the signal of primary user $i$. Rate-splitting effect on this region is determined in the following theorem.

\begin{theorem}
\label{th:R^r_g}
For $Z\in\mathcal{G}(P_1,P_2,P_s)$ and $i\in\{1,2\}$, an achievable rate region $\mathcal{R}^r_{i}(Z)$ coincides on $\mathcal{R}^r_{ig}$ if and only if $\lambda=0$ and
\begin{equation}
I(Y_s;W|X_i)\leq{I(Y_p;W|X_1X_2)}
\end{equation}
or equivalently,
\begin{equation}
\label{eq:cond3}
g_s^s N_0\leq g_s^p(g_j^sP_j+N_0),\quad j\in\{1,2\},j\neq i.
\end{equation}
\end{theorem}

\begin{proof}
Please refer to Appendix \ref{App:R^r_g}
\end{proof}

Hence, if inequality \eqref{eq:cond3} is satisfied, $\mathcal{R}^r_{ig}$ is obtained without rate-splitting, specifically, when $\lambda=0$.

\begin{figure}[ht]
  \centering
  \subfloat[The overall achievable rate region $\mathcal{R}^r_{1g}$ is obtained when the whole secondary signal is decodable by the primary receiver. $\mathcal{R}^r_{1g}$ is shown in blue.]{\label{fig:R^r_g_C1}\includegraphics[width=0.3\textwidth]{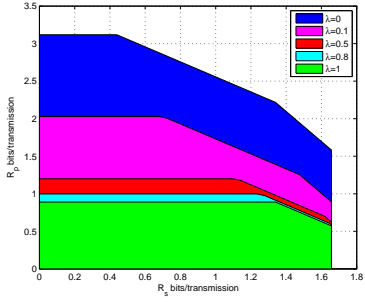}}\\                
  \subfloat[Rate-splitting affects the achievable rate region. $\mathcal{R}^r_{1g}$ is shown in blue and $\mathcal{R}^r_{1}(Z)$ is shown in green for $\lambda=0$, yellow for $\lambda=0.1$ and red for $\lambda=1$.]{\label{fig:R^r_g_C2}\includegraphics[width=0.3\textwidth]{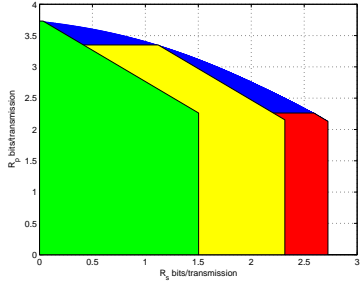}}
  
  \caption{Performance of rate-splitting by the secondary link when it can decode the signal of primary user $1$.}
  \label{fig:RS_performance2}
\end{figure}

Fig. \ref{fig:RS_performance2} shows the performance of rate-splitting under same power setup used with Fig. \ref{fig:RS_performance1}, where it is assumed that the secondary receiver can decode the signal of primary user $1$. In Fig. \ref{fig:R^r_g_C1} the achievable rate region $\mathcal{R}^r_{1g}$ coincides on $\mathcal{R}^r_{1}(Z)$ when inequality \eqref{eq:cond3} is satisfied. The parameters for this scenario are $g_1^p=5.5303$, $g_2^p=4.2865$, $g_1^s=0.6542$, $g_2^s=0.8121$, $g_s^p=3.9334$ and $g_s^s=8.1575$.

In Fig. \ref{fig:R^r_g_C2} the opposite scenario is considered where inequality \eqref{eq:cond3} is not satisfied. It is obvious that the overall rate region $\mathcal{R}^r_{1g}$ is obtained by varying $\lambda$ from $0$ to $1$ as a consequence of the fact that rate regions corresponding to different values of $\lambda$ do not include one another if inequality \eqref{eq:cond3} is not satisfied. The channel gains for Fig. \ref{fig:R^r_g_C2} are $g_1^p=9.566$, $g_2^p=14.5045$, $g_1^s=0.0808$, $g_2^s=0.2894$, $g_s^p=0.7032$ and $g_s^s=16.6226$.

Consequently, the achievable rate region $\mathcal{R}_{ig}$ coincides on $\mathcal{R}_{ig}(Z)$ at $\lambda=0$ if and only if \eqref{eq:cond3} is satisfied.

\subsection{On Decoding One Primary Signal}

In Subsection \ref{subsec:ARR_CRSp} we introduce an achievable rate-region for the DMC case assuming that the signal of one primary transmitters has to be reliably decoded by the secondary receiver. Although this may impose a constraint on the range of achievable sum rates by the primary network, we showed in Theorem \ref{th:inclusion} and Corollary \ref{cor:Inclusion} that there exists a condition for which this constraint only enhances the achievable rates for the secondary link without degrading the range of achievable rates by the primary network. This condition is called PDC. When applying this condition to the given Gaussian channel the PDC would be: If for all $Z\in\mathcal{G}(P_1,P_2,P_s)$ $I(Y_p;X_i|W)\leq I(Y_s;X_i|UW)$ then $\mathcal{R}^o_g\subseteq\mathcal{R}^r_{ig}$. Equivalently, the following inequality must hold, 
\begin{multline}
\label{eq:PDCG_1}
\tau \left(\frac{g_i^pP_i}{g_s^p\lambda P_s+g_j^pP_j+N_0}\right)\leq\tau\left(\frac{g_i^sP_i}{g_j^sP_j+N_0}\right),\\ \forall \lambda: 0\leq\lambda\leq 1,\quad j\neq i,\quad i,j\in\{1,2\}.
\end{multline}

But since $I(Ys;X_i|UW)$ does not depend on $\lambda$, then a necessary and sufficient condition to have \eqref{eq:PDCG_1} satisfied is

\begin{equation}
\label{eq:PDCG_2}
\frac{g_i^p}{g_j^pP_j+N_0}\leq\frac{g_i^s}{g_j^sP_j+N_0},\quad j\neq i, \quad i,j\in\{1,2\}.
\end{equation}

We call inequality \eqref{eq:PDCG_2} primary decodability condition for Gaussian channel (PDCG).

Fig. \ref{fig:JournalInclusion} shows a scenario for which three rate regions are obtained: $\mathcal{R}^o_g$, $\mathcal{R}^r_{1g}$ and $\mathcal{R}^r_{2g}$. It is clear that $\mathcal{R}^o_g\subseteq\mathcal{R}^r_{1g}$ meaning that primary user $1$ satisfies the PDCG described in \eqref{eq:PDCG_2}, whereas primary user $2$ does not. By decoding the signal of primary user $1$ at the secondary receiver, the range of achievable primary rates in $\mathcal{R}^o_g$ remains the same for $\mathcal{R}^r_{1g}$ while the secondary link can achieve higher rate at a given primary rate in $\mathcal{R}^r_{1g}$ than in $\mathcal{R}^o_g$. The power setup used to produce this figure is the same as that of Fig. \ref{fig:RS_performance1} and the channel gains are $g_1^p=0.3413$, $g_2^p=10.2047$, $g_1^s=0.2821$, $g_2^s=0.3782$, $g_s^p=0.2495$ and $g_s^s=6.3337$.

\begin{figure}
	\centering
		\includegraphics[width=0.40\textwidth]{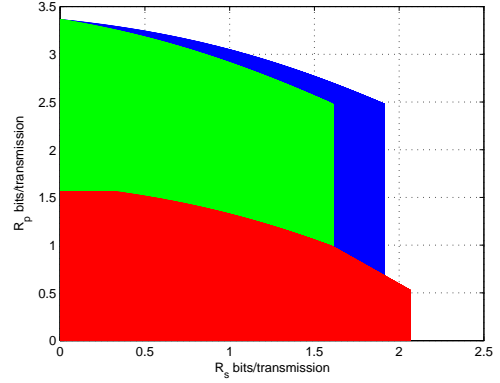}
	\caption{Achievable rate regions for the Gaussian channel. $\mathcal{R}^o_g$ is shown in green, $\mathcal{R}^r_{1g}$ in blue and $\mathcal{R}^r_{2g}$ in red.}
	\label{fig:JournalInclusion}
\end{figure}

\begin{figure}[ht]
	\centering
		\includegraphics[width=0.40\textwidth]{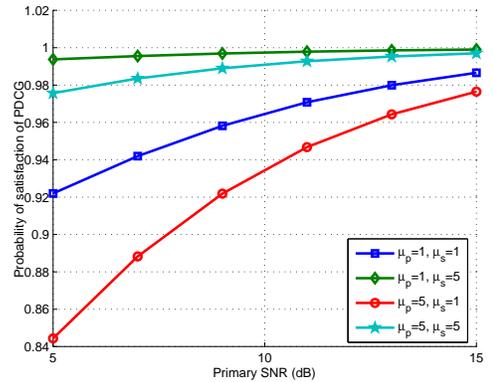}
	\caption{Probability of finding at least one primary user that satisfies the PDCG}
	\label{fig:ProbPDC}
\end{figure}

Note that, a primary user that satisfies PDCG does not always exist, so we evaluate the probability of PDCG as the probability of finding at least one primary user satisfying \eqref{eq:PDCG_2}. We assume $N_0=1$ unit power and  $g_1^s$ and $g_2^s$ are i.i.d. exponentially distributed with mean $\mu_s$, whereas $g_1^p$ and $g_2^p$ are i.i.d. exponentially distributed with mean $\mu_p$, where $g_1^s$, $g_2^s$, $g_1^p$ and $g_2^p$ are mutually independent. A closed form formula for the probability of PDCG is difficult to obtain, so we evaluate it numerically by generating $10^7$ different values for each channel gain element and calculating the average number of times at which neither primary user satisfies \eqref{eq:PDCG_2} at a given $P_1$ and $P_2$, then by subtracting it from $1$ we get a numerical estimate for the probability of PDCG. A simulation has been done in which we assume that $\frac{P_1}{N_0}=\frac{P_2}{N_0}=\snr_p$. We vary $\snr_p$ and evaluate the corresponding probability of PDCG. This simulation is done for the following pairs of $(\mu_p,\mu_s)$: $(1,1)$, $(1,5)$, $(5,1)$ and $(5,5)$. The result is shown in Fig. \ref{fig:ProbPDC}, where it is obvious that the probability of PDCG increases with $\snr_p$, and that the increase in $\mu_s$ yields more increase in probability of PDCG.

\section{Conclusion}
\label{sec:conc}
In this work we established an achievable rate region for a primary multiple access network coexisting with a secondary link that comprises one transmitter and a corresponding receiver. The achievable rate regions are obtained for the sum primary rate versus the secondary rate. We first considered DMC where the secondary link employs rate-splitting, and established two types of achievable rate regions: one type is when the secondary receiver cannot decode any of the primary signals, whereas the second is when the secondary is able to decode the signal of only one primary transmitter. The overall achievable rate region is the union of those two types of regions. Moreover, we showed that there exists a case for which allowing the secondary receiver to decode a primary signal results in an achievable rate region that includes the achievable rate region obtained when the secondary receiver does not decode the primary signal. Then, we investigated the performance of rate-splitting in the Gaussian channel where it was found that rate-splitting by the secondary user is useless when the channel between the secondary transmitter and the primary receiver supports larger rate than the channel between the two secondary nodes. Furthermore, on decoding the signal of a primary transmitter at the secondary receiver, a necessary and sufficient condition has been provided to allow the secondary user decode the primary signal without reducing the range of achievable primary sum rates but only increases the range of achievable secondary rates. Finally, we showed numerically that the probability of finding at least one primary user that satisfies this condition increases with the signal to noise ratio of the primary users.

\appendices
\section{Proof of Theorem \ref{th:ARR_CRS}}
\label{App:ARR_CRRS}

It is sufficient to show that there exists at least one code for which if the rate tuple $(R_1, R_2, S, T)$ satisfies \eqref{eq:g_relationsF}-\eqref{eq:g_relationsL} then the rate tuple is achievable. We use the following random code.
 
\subsection{Random Code Generation}
A random code $\mathcal{C}$ is generated as follows. Let $\textbf{q}=(q^{(1)},\cdots,q^{(n)})$ be a random i.i.d sequence of $\mathcal{Q}^n$, $\textbf{u}_k=(u_k^{(1)},\cdots,u_k^{(n)})$, $k\in\mathcal{L}_s$ a sequence of random variables of $\mathcal{U}^n$ that are i.i.d given $\textbf{q}$. Moreover, $\textbf{u}_k$ and $\textbf{u}_{k'}$ are independent $\forall k\neq k'$, $k,k'\in\mathcal{L}_s$. Similarly, generate $\textbf{w}_l$, $l\in\mathcal{N}_s$, $\textbf{x}_{1i}$, $i\in\mathcal{M}_1$ and $\textbf{x}_{2j}$, $j\in\mathcal{M}_2$.
\subsection{Encoding}
For primary user $1$ to send a message $i\in\mathcal{M}_1$, it sends $\textbf{x}_{1i}$. Similarly, for primary user $2$ to send a message $j\in\mathcal{M}_2$, it sends $\textbf{x}_{2j}$. For the secondary user to send a message $kl\in\mathcal{L}_s\times\mathcal{N}_s$, it sends $f^n(\textbf{u}_k\textbf{w}_l|\textbf{q})=\left(f^{(1)}(u^{(1)}_kw^{(1)}_l|q^{(1)}), \cdots, f^{(n)}(u^{(n)}_kw^{(n)}_l|q^{(n)}) \right)$, where $\textbf{q}$ is known at the transmitters.

\subsection{Decoding: Jointly-Typical Decoding}
We use the concept of jointly typical sequences and the properties of typical sets introduced in Chapter 15 of \cite{IT} to implement the decoding functions. Let $A_{\epsilon}^{(n)}$ denote the set of typical $(\textbf{q},\textbf{x}_1,\textbf{x}_2,\textbf{w}_l,\textbf{y}_p)$ sequences, then the primary receiver decides $ijl$ if $(\textbf{q},\textbf{x}_{1i},\textbf{x}_{2j},\textbf{w}_l,\textbf{y}_p)\in A_{\epsilon}^{(n)}$. Also, for $B_{\epsilon}^{(n)}$ is the set of typical $(\textbf{q},\textbf{u},\textbf{w},\textbf{y}_s)$ sequences, the secondary receiver decides $kl$ if $(\textbf{q},\textbf{u}_k,\textbf{w}_l,\textbf{y}_s)\in B_{\epsilon}^{(n)}$.

\subsection{Probability of Error Analysis}
By the symmetry of the random code generation, the conditional probability of error does not depend on the transmitted messages. Hence, the conditional probability of error is the same as the average probability of error. So, let $ijkl=1111$ are sent. An error occurs if the transmitted codewords are not typical with the received sequences.
\subsubsection{For the Primary Receiver}
Let the event
\begin{equation*}
\label{eq:E_p}
E_p(ijl)=\left\{(\textbf{q},\textbf{x}_{1i},\textbf{x}_{2j},\textbf{w}_l,\textbf{y}_p)\in A_{\epsilon}^{(n)}\right\}, 
\end{equation*}
hence the probability of error averaged over the random code $\mathcal{C}$ is
\begin{equation*}
\bar{P}e_p^o=P\left(E_p^c(111)\bigcup \cup_{ijl\neq 111} E_p(ijl)\right),
\end{equation*}
where $E_p^c(111)$ denotes the complement of $E_p(111)$. Using union bound we have
\begin{equation*}
\begin{split}
\bar{P}e_p^o & \leq  P\left(E_p^c(111)\right)+P\left(\cup_{ijl\neq 111}E_p(ijl)\right) \\
						 & \leq  P\left(E_p^c(111)\right)+(M_1-1)P(E_p(211)) \\
						 & \quad +(M_2-1) P(E_p(121))+(N_s-1) P(E_p(112))\\
						 & \quad +(M_1-1)(M_2-1) P(E_p(221))\\
						 & \quad +(M_1-1)(N_s-1) P(E_p(212))\\
						 & \quad +(M_2-1)(N_s-1) P(E_p(122))\\
						 & \quad +(M_1-1)(M_2-1)(N_s-1) P(E_p(222)).
\end{split}
\end{equation*}

From the properties of jointly typical sequences \cite{IT}, $P(E_p^c(111))\rightarrow\epsilon$ as $n\rightarrow\infty$, and 
\begin{equation*}
\begin{split}
P(E_p(211))&=2^{-n(H(X_1|Q)-H(X_1|X_2WY_pQ))+6\epsilon}\\
           &=2^{-n(I(X_1;X_2WY_p|Q))+6\epsilon}\\
           &=2^{-n(I(Y_p;X_1|WX_2Q))+6\epsilon},
\end{split}
\end{equation*}
where the last equality holds from the assumption that $X_1$, $X_2$, $U$ and $W$ are independent and conditionally independent given $Q$. Similarly for other $E_p(ijl\neq{111})$ and applying Equations \eqref{eq:R^o_F}-\eqref{eq:R^o_L} we get
\begin{equation*}
\begin{split}
\bar{P}e_p^o &\leq 2^{-n(I(Y_p;X_1|WX_2Q)-R_1+\eta-6\epsilon)}\\
						 & \quad +2^{-n(I(Y_p;X_2|WX_1Q)-R_2+\eta-6\epsilon)}\\
						 & \quad +2^{-n(I(Y_p;W|X_1X_2Q)-T+\eta-6\epsilon)}\\
						 & \quad +2^{-n(I(Y_p;X_1X_2|WQ)-(R_1+R_2)+\eta-6\epsilon)}\\
						 & \quad +2^{-n(I(Y_p;WX_1|X_2Q)-(T+R_1)+\eta-6\epsilon)}\\
						 & \quad +2^{-n(I(Y_p;WX_2|X_1Q)-(T+R_2)+\eta-6\epsilon)}\\
						 & \quad +2^{-n(I(Y_p;X_1X_2W|Q)-(T+R_1+R_2)+\eta-6\epsilon)}.
\end{split}
\end{equation*} 
Thus if \eqref{eq:g_relationsF}-\eqref{eq:g_relationsMp} are satisfied, $\bar{P}e_p^o\rightarrow \epsilon$ as $n\rightarrow\infty$. 

\subsubsection{For the Secondary Receiver}
Let the event
\begin{equation*}
E_s(kl)=\left\{(\textbf{q},\textbf{u}_k,\textbf{w}_l,\textbf{y}_s)\in B_{\epsilon}^{(n)}\right\}
\end{equation*}
hence the probability of decoding error averaged over the random code $\mathcal{C}$ is
\begin{equation*}
\bar{P}e_s^o=P\left(E_s^c(11)\bigcup \cup_{kl\neq 11} E_p(kl)\right),
\end{equation*}
where $E_s^c(11)$ denotes the complement of $E_s(11)$. Using union bound we have
\begin{equation*}
\begin{split}
\bar{P}e_s^o \leq & P(E_s^c(11))+(L_s-1)P(E_s(21))\\
								  &+(N_s-1)P(E_s(12))\\
								  &+(L_s-1)(N_s-1)P(E_s(22)).
\end{split}
\end{equation*}
Since $P(E_s^c(11))\rightarrow\epsilon$ as $n\rightarrow\infty$, then
\begin{equation*}
\begin{split}
\bar{P}e_s^o\leq & 2^{-n(I(Y_s;U|WQ)-S+\eta-6\epsilon)}\\
                 & +2^{-n(I(Y_s;W|UQ)-T+\eta-6\epsilon)}\\
                 & +2^{-n(I(Y_s;UW|Q)-(S+T)+\eta-6\epsilon)}
\end{split}
\end{equation*}
So, if \eqref{eq:g_relationsMs}-\eqref{eq:g_relationsL} are satisfied, $\bar{P}e_s^o\rightarrow\epsilon$ as $n\rightarrow\infty$. 

This concludes the proof.

\section{Proof of Theorem \ref{th:inclusion}}
\label{app:RRproof}
\subsection{Sufficiency Part}
Suppose \eqref{eq:cond} is satisfied, we use Fig. \ref{fig:RRproof} to prove that $\mathcal{R}^o(Z)\subseteq\mathcal{R}_i^r(Z)$. It is sufficient to show that $R_p^{A^o}= R_p^{A^r}$, $R_s^{B^o}\leq R_s^{B^r}$, $R_s^{D^o}\leq R_s^{F^r}$ and that lines $2R_s+R_p=\rho_{2p}^r$ and $R_s+R_p=\rho_{sp}^o$ intersect at a point $(R_s^*,R_p^*)$ for which $R_s^*\geq R_s^{D^o}$, i.e., the intersection between the two lines is outside $\mathcal{R}^o(Z)$. Consider the primary user whose signal is not decodable at the secondary receiver is indexed by $j$, $j\in\{1,2\}$ and $i\neq j$.

\begin{figure}
	\centering
		\includegraphics[width=0.28\textwidth]{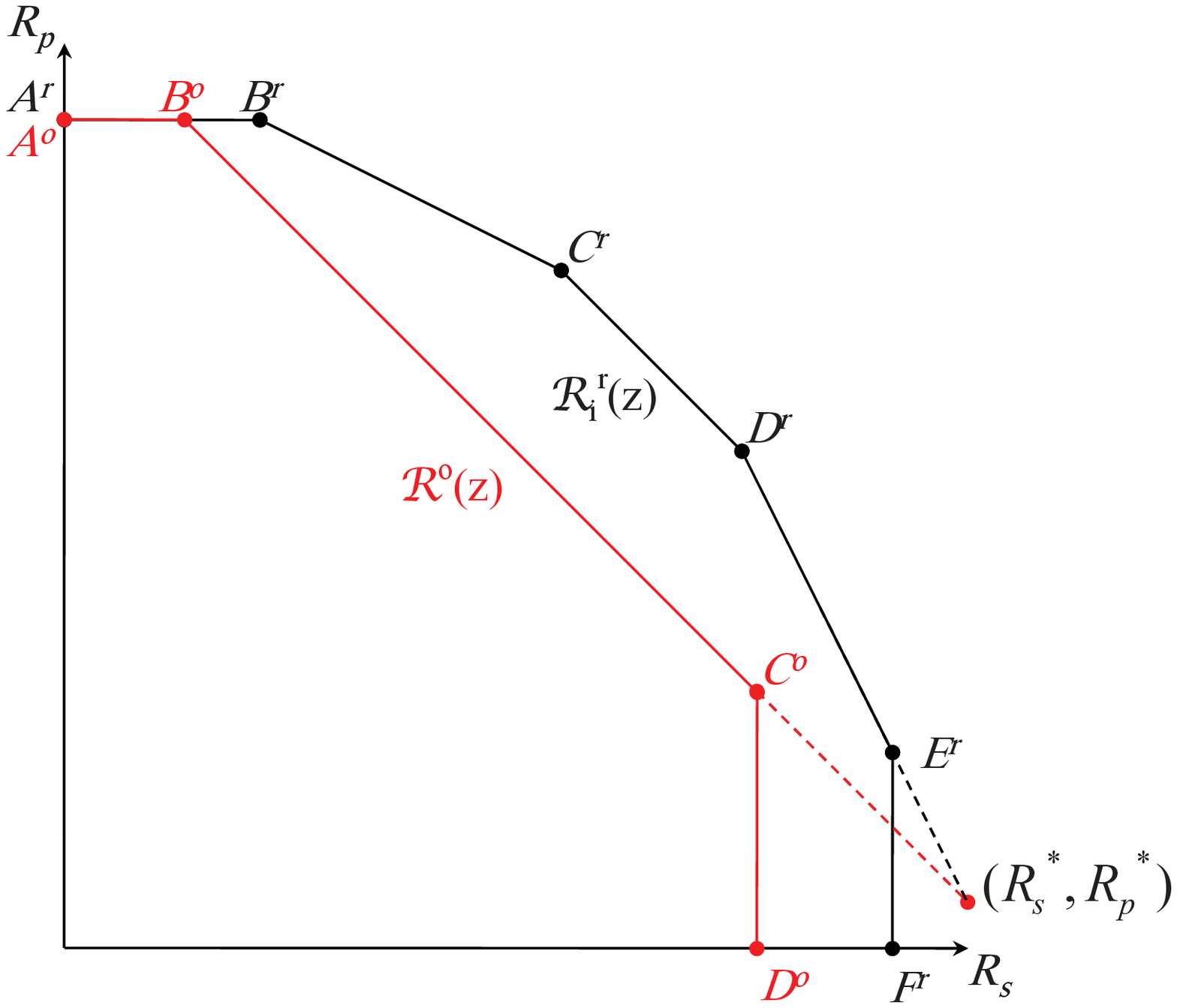}
	\caption{Regions $\mathcal{R}_i^r(Z)$ and $\mathcal{R}^o(Z)$ when $I(Y_p;X_i|WQ)\leq I(Y_s;X_i|UWQ)$.}
	\label{fig:RRproof}
\end{figure}

\subsubsection{Proof of $R_p^{A^o}= R_p^{A^r}$}
From the analysis of the channels $C_{RS}$ and $C_{RS}^p$ in Section \ref{sec:ARR} we have
\begin{eqnarray*}
R_p^{A^o}=&I(Y_p;X_1X_2|WQ),\\
R_p^{A^r}=&I(Y_p;X_j|WX_iQ)+\sigma_p^*.
\end{eqnarray*}
From \eqref{eq:cond}, $\sigma_p^*=I(Y_p,X_i|WQ)$. Therefore,
\begin{equation*}
R_p^{A^r}=I(Y_p;X_1X_2|WQ)=R_p^{A^o}.
\end{equation*}

\subsubsection{Proof of $R_s^{B^o}\leq R_s^{B^r}$}
From the proof of Theorem \ref{th:R^o(Z)}
\begin{equation}
\label{eq:R_s^Bo}
R_s^{B^o}=I(Y_s;U|WQ)+\min\{\overbrace{I(Y_p;W|Q)}^{o_1},\overbrace{I(Y_s;W|Q)}^{o_2}\},
\end{equation}
and from the proof of Theorem \ref{th:R^r(Z)}
\begin{equation*}
\begin{split}
R_s^{B^r}=& I(Y_s;U|WX_iQ)-[I(Y_p;X_i|WQ)\\
          &-I(Y_s;X_i|WQ)]^++\min\{I(Y_p;W|Q),\\
          & I(Y_s;W|Q)+[I(Y_s;X_i|WQ)-I(Y_p;X_i|WQ)]^+,\\
          &I(Y_s;W|X_iQ)\}.
\end{split}
\end{equation*}
\paragraph{If $I(Y_p;X_i|WQ)\leq I(Y_s;X_i|WQ)$}
\begin{equation}
\label{eq:R_s^Br1}
\begin{split}
R_s^{B^r}=& I(Y_s;U|WX_iQ)+\min\{\overbrace{I(Y_p;W|Q)}^{\nu_1},\\
          & \overbrace{I(Y_s;W|Q)+I(Y_s;X_i|WQ)-I(Y_p;X_i|WQ)}^{\nu_2},\\
          & \overbrace{I(Y_s;W|X_iQ)}^{\nu_3}\}.
\end{split}
\end{equation}
Note that, $\nu_1=o_1$.
\begin{itemize}
\item If $o_1\leq o_2$ in \eqref{eq:R_s^Bo}
\end{itemize}
\begin{equation*}
R_s^{B^o}=I(Y_s;U|WQ)+I(Y_p;W|Q),
\end{equation*}
\begin{equation*}
\begin{split}
R_s^{B^r}&=I(Y_s;U|WX_iQ)+I(Yp;W|Q)\\
         &\geq R_s^{B^o}. 
\end{split}
\end{equation*}

\begin{itemize}
	\item If $o_2\leq o_1$ in \eqref{eq:R_s^Bo}
\end{itemize}
\begin{equation*}
\begin{split}
R_s^{B^o}=& I(Y_s;U|WQ)+I(Y_s;W|Q)\\
         =& I(Y_s;UW|Q).
\end{split}
\end{equation*}
When $\nu_1= \min\{\nu_1,\nu_2,\nu_3\}$ in \eqref{eq:R_s^Br1}, then
\begin{equation*}
\begin{split}
R_s^{B^r}&= I(Y_s;U|WX_iQ)+\overbrace{I(Y_p;W|Q)}^{\geq o_2}\\
         & \geq R_s^{B^o}.
\end{split}
\end{equation*} 
When $\nu_2= \min\{\nu_1,\nu_2,\nu_3\}$ in \eqref{eq:R_s^Br1}, then
\begin{equation*}
\begin{split}
R_s^{B^r} & =I(Y_s;U|WX_iQ)+I(Y_s;W|Q)+I(Y_s;X_i|WQ)\\
          & \quad -I(Y_p;X_i|WQ)\\
          & \geq  I(Y_s;U|WX_iQ)+I(Y_s;W|Q)\\
          & \geq  R_s^{B^o}.
\end{split}
\end{equation*}
When $\nu_3= \min\{\nu_1,\nu_2,\nu_3\}$ in \eqref{eq:R_s^Br1}, then
\begin{equation*}
\begin{split}
R_s^{B^r}&=I(Y_s;U|WX_iQ)+I(Y_s;W|X_iQ)\\
         &=I(Y_s;UW|X_iQ)\\
         & \geq R_s^{B^o}.
\end{split}
\end{equation*}

\paragraph{If $I(Y_s;X_i|WQ)\leq I(Y_p;X_i|WQ)$}

\begin{equation}
\label{eq:R_s^Br2}
\begin{split}
R_s^{B^r}=& I(Y_s;U|WX_iQ)+I(Y_s;X_i|WQ)-I(Y_p;X_i|WQ)\\
          & +\min\{\overbrace{I(Y_p;W|Q)}^{\nu_4},\overbrace{I(Y_s;W|Q)}^{\nu_5}\}.
\end{split}
\end{equation}
Note that, $o_1=\nu_4$ and $o_2=\nu_5$.
\begin{itemize}
	\item If $o_1\leq o_2$ in \eqref{eq:R_s^Bo}
\end{itemize}
\begin{equation*}
\begin{split}
R_s^{B^o}=I(Y_s;U|WQ)+I(Y_p;W|Q),
\end{split}
\end{equation*}

\begin{equation*}
\begin{split}
R_s^{B^r} & = I(Y_s;UX_i|WQ)-I(Y_p;X_i|WQ)+I(Y_p;W|Q)\\
          & = I(Y_s;U|WQ)+I(Y_p;W|Q)\\
          & \quad +\overbrace{I(Y_s;X_i|UWQ)-I(Y_p;X_i|WQ)}^{\geq 0 \mbox{ from \eqref{eq:cond}}}\\
          & \geq R_s^{B^o}.    
\end{split}
\end{equation*}

\begin{itemize}
	\item If $o_2 \leq o_1$ in \eqref{eq:R_s^Bo}
\end{itemize}
Proof follows exactly as the case of $o_1 \leq o_2$.

\subsubsection {Proof of $R_s^{F^r}\geq R_s^{D^o}$}
\begin{equation*}
\begin{split}
R_s^{F^r}=& I(Y_s;U|WX_iQ)+\min\{I(Y_s;W|X_iQ),\\
          & I(Y_p;W|X_1X_2Q)\}.
\end{split}
\end{equation*}
\begin{equation*}
\begin{split}
R_s^{D^o}=& I(Y_s;U|WQ)+\min\{I(Y_s;W|Q),\\
          & I(Y_p;W|X_1X_2Q)\}.
\end{split}
\end{equation*}
It is obvious that each term in $R_s^{F^r}$ is greater than or equal to its corresponding term in $R_s^{D^o}$. Hence, $R_s^{F^r}\geq R_s^{D^o}$.

\subsubsection {Proof of the intersection point between the two lines $2R_s+R_p=\rho_{2p}^r$ and $R_s+R_p=\rho_{sp}^o$ occurs at a point $(R_s^*,R_p^*)$ where $R_s^*\geq R_s^{D^o}$}

The secondary rate of the intersection point is $R_s^*=\rho_{2p}^r-\rho_{sp}^o$. From Theorems \ref{th:R^o(Z)} and \ref{th:R^r(Z)}
\begin{equation}
\label{eq:R_s^DoG}
R_s^{D^o}=I(Y_s;U|WQ)+\sigma^*,
\end{equation}
\begin{equation}
\label{eq:R_s^*G}
\begin{split}
R_s^*=& 2I(Y_s;U|WX_iQ)+2\sigma_s^*+I(Y_p;X_j|WX_iQ)\\
      & -[\sigma_s^*-I(Y_p;W|X_iQ)]^++\min\{I(Y_s;X_i|WQ),\\
      & I(Y_s;WX_i|Q)-\sigma_s^*,I(Y_p;X_i|Q)+[I(Y_p;W|X_iQ)\\
      & -\sigma_s^*]^+,I(Y_p;X_i|WQ)\}-I(Y_p;X_1X_2|WQ)\\
      & -I(Y_s;U|WQ)-\min\{I(Y_s;W|Q),I(Y_p;W|Q)\}.
\end{split}
\end{equation}
Hence, it is required to show that $R_s^*\geq R_s^{D^o}$
\paragraph{If $\sigma_s^*=I(Y_s;W|X_iQ)\leq I(Y_p;W|X_1X_2Q)$}

\begin{itemize}
	\item If $I(Y_s;W|X_iQ)\leq I(Y_p;W|X_iQ)$,
\end{itemize}
from \eqref{eq:R_s^DoG} and \eqref{eq:R_s^*G} we have
\begin{equation}
\label{eq:R_s^Do1}
\begin{split}
R_s^{D^o}=I(Y_s;U|WQ)+I(Y_s;W|Q)=I(Y_s;UW|Q),
\end{split}
\end{equation}

\begin{equation}
\label{eq:R_s^*1}
\begin{split}
R_s^*=& 2 I(Y_s;U|WX_iQ)+2 I(Y_s;W|X_iQ)-I(Y_s;U|WQ)\\
      & +I(Y_p;X_j|WX_iQ)+\min\{\overbrace{I(Y_s;X_i|Q)}^{\nu_6},\\
      & \overbrace{I(Y_p;X_i|Q)+I(Y_p;W|X_iQ)-I(Y_s;W|X_iQ)}^{\nu_7},\\
      & \overbrace{I(Y_p;X_i|WQ)}^{\nu_8}\}-I(Y_p;X_1X_2|WQ)-\min\{\nu_4,\nu_5\}.
\end{split}
\end{equation}

When $\nu_6=\min\{\nu_6,\nu_7,\nu_8\}$ in \eqref{eq:R_s^*1}, then 

\begin{equation*}
\begin{split}
R_s^*&= I(Y_s;U|WX_iQ)-I(Y_s;U|WQ)+\overbrace{I(Y_s;UW|Q)}^{=R_s^{D^o}}\\
     & \quad +I(Y_s;W|X_iQ)-\min\{\nu_4,\nu_5\}+I(Y_p;X_j|WQ)\\
     & \quad +\overbrace{I(Y_s;X_i|UWQ)}^{\geq I(Y_p;X_i|WQ) \mbox{ from \eqref{eq:cond}}}-I(Y_p;X_1X_2|WQ)\\
     & \geq R_s^{D^o}.
\end{split}
\end{equation*}

When $\nu_7=\min\{\nu_6,\nu_7,\nu_8\}$ in \eqref{eq:R_s^*1}, then

\begin{equation*}
\begin{split}
R_s^*&= I(Y_s;U|WX_iQ)-I(Y_s;U|WQ)+\overbrace{I(Y_s;UW|X_iQ)}^{\geq R_s^{D^o}}\\
     & \quad +\nu_4-\min\{\nu_4,\nu_5\}\\
     & \geq R_s^{D^o}.
\end{split}
\end{equation*}

When $\nu_8=\min\{\nu_6,\nu_7,\nu_8\}$ in \eqref{eq:R_s^*1}, then

\begin{equation*}
\begin{split}
R_s^*&= I(Y_s;U|WX_iQ)-I(Y_s;U|WQ)+\overbrace{I(Y_s;UW|X_iQ)}^{\geq R_s^{D^o}}\\
     & \quad +I(Y_s;W|X_iQ)-\min\{\nu_4,\nu_5\}\\
     & \geq R_s^{D^o}.
\end{split}
\end{equation*}

\begin{itemize}
	\item If $I(Y_s;W|X_iQ)\geq I(Y_p;W|X_iQ)$,
\end{itemize}
$R_s^{D^o}$ will remain the same as \eqref{eq:R_s^Do1} and $R_s^*$ will be given by
\begin{equation}
\label{eq:R_s^*2}
\begin{split}
R_s^*=&2I(Y_s;U|WX_iQ)+2 I(Y_s;W|X_iQ)-I(Y_s;U|WQ)\\
      &+I(Y_p;X_j|WX_iQ)+I(Y_p;W|X_iQ)\\
      &-I(Y_s;W|X_iQ)+\min\{\overbrace{I(Y_s;X_i|Q)}^{\nu_9},\overbrace{I(Y_p;X_i|Q)}^{\nu_{10}}\}\\
      &-I(Y_p;X_1X_2|WQ)-\min\{\nu_4,\nu_5\}.
\end{split}
\end{equation}

When $\nu_9=\min\{\nu_9,\nu_{10}\}$ in \eqref{eq:R_s^*2}, then

\begin{equation*}
\begin{split}
R_s^* &=      I(Y_s;U|WX_iQ)-I(Y_s;U|WQ)+\overbrace{I(Y_s;UW|Q)}^{=R_s^{D^o}}\\
      & \quad +I(Y_p;W|X_iQ)-\min\{\nu_4,\nu_5 \}+ I(Y_s;X_i|UWQ)\\
      & \quad +I(Y_p;X_j|WX_iQ)-I(Y_p;X_1X_2|WQ)\\
      & \geq R_s^{D^o}.
\end{split}
\end{equation*}

When $\nu_{10}=\min\{\nu_9,\nu_{10}\}$ in \eqref{eq:R_s^*2}, then

\begin{equation*}
\begin{split}
R_s^* &=       I(Y_s;U|WX_iQ)-I(Y_s;U|WQ)+\overbrace{I(Y_s;UW|X_iQ)}^{\geq R_s^{D^o}}\\
      & \quad +I(Y_p;W|X_iQ)-\min\{\nu_4,\nu_5\} \\
      & \geq R_s^{D^o}
\end{split}
\end{equation*}

\paragraph{If $\sigma_s^*=I(Y_p;W|X_1X_2Q)\leq I(Y_s;W|X_iQ)$}
from \eqref{eq:R_s^DoG} and \eqref{eq:R_s^*G} we have
\begin{equation}
\label{eq:R_s^Do3}
R_s^{D^o}=I(Y_s;U|WQ)+\min\{\overbrace{I(Y_s;W|Q)}^{o_2},\overbrace{I(Y_p;W|X_1X_2Q)}^{o_3}\},
\end{equation}

\begin{equation}
\label{eq:R_s^*3}
\begin{split}
R_s^*= &2 I(Y_s;U|WX_iQ)-I(Y_s;U|WQ)+I(Y_p;W|X_iQ)\\
       &+I(Y_p;X_j|WX_iQ)+I(Y_p;W|X_1X_2Q)\\
       &+\min\{\overbrace{I(Y_p;X_i|Q)}^{\nu_{10}},\overbrace{I(Y_s;X_i|WQ)}^{\nu_{11}},\\
       &\overbrace{I(Y_s;WX_i|Q)-I(Y_p;W|X_1X_2Q)}^{\nu_{12}}\}-\min\{\nu_4,\nu_5\} \\
       &-I(Y_p;X_1X_2|WQ)
\end{split}
\end{equation}

\begin{itemize}
	\item If $o_2\leq o_3$ in \eqref{eq:R_s^Do3},
\end{itemize}
\begin{equation*}
R_s^{D^o}=I(Y_s;U|WQ)+I(Y_s;W|Q)=I(Y_s;UW|Q).
\end{equation*}
When $\nu_{10}=\min\{\nu_{10},\nu_{11},\nu_{12}\}$ in \eqref{eq:R_s^*3}, then
\begin{equation*}
\begin{split}
R_s^*& =I(Y_s;U|WX_iQ)-I(Y_s;U|WQ)+\nu_4\\
     &  \quad -\min\{\nu_4,\nu_5\}+\overbrace{I(Y_s;U|WX_iQ)+o_3}^{\geq R_s^{D^o}}.
\end{split}
\end{equation*}
Since $o_2\leq o_3$, then $\nu_{11}$ cannot be smaller than $\nu_{12}$.
When $\nu_{12}=\min\{\nu_{10},\nu_{11},\nu_{12}\}$, then
\begin{equation*}
\begin{split}
R_s^*&=I(Y_s;U|WX_iQ)-I(Y_s;U|WQ)+I(Y_p;W|X_iQ)\\
     & \quad -\min\{\nu_4,\nu_5\}+\overbrace{I(Y_s;UW|Q)}^{=R_s^{D^o}}+I(Y_s;X_i|UWQ)\\
     & \quad +I(Y_p;X_j|WX_iQ)-I(Y_p;X_1X_2|WQ)\\
     & \geq R_s^{D^o}.
\end{split}
\end{equation*}

\begin{itemize}
	\item If $o_2\geq o_3$
\end{itemize}

\begin{equation*}
R_s^{D^o}=I(Y_s;U|WQ)+I(Y_p;W|X_1X_2Q).
\end{equation*}

When $\nu_{10}=\min\{\nu_{10},\nu_{11},\nu_{12}\}$ in \eqref{eq:R_s^*3}, then
\begin{equation*}
\begin{split}
R_s^*&=I(Y_s;U|WX_iQ)-I(Y_s;U|WQ)\\
     & \quad +\overbrace{I(Y_s;U|WX_iQ)+I(Y_p;W|X_1X_2Q)}^{\geq R_s^{D^o}}\\
     & \geq R_s^{D^o}.
\end{split}
\end{equation*}

When $\nu_{11}=\min\{\nu_{10},\nu_{11},\nu_{12}\}$ in \eqref{eq:R_s^*3}, then
\begin{equation*}
\begin{split}
R_s^*&=I(Y_s;U|WX_iQ)-I(Y_s;U|WQ)+I(Y_s;X_i|UWQ)\\
     & \quad +I(Y_p;X_j|WX_iQ)-I(Y_p;X_1X_2|WQ)-\nu_4\\
     & \quad +I(Y_p;W|X_iQ)+\overbrace{I(Y_s;U|WQ)+I(Y_p;W|X_1X_2Q)}^{=R_s^{D^o}}\\
     & \geq R_s^{D^o}.
\end{split}
\end{equation*}
Since $o_2 \geq o_3$ then $\nu_{12}$ cannot be smaller than $\nu_{11}$.

\subsection {Necessity Part}

Suppose $\mathcal{R}^o(Z)\subseteq\mathcal{R}_i^r(Z)$ then $R_s^{A^o}$ must be not larger than $R_s^{A^r}$ which necessitates the satisfaction of \eqref{eq:cond}.

This concludes the proof.

\section{Proof of Theorem \ref{th:R^o_g}}
\label{App:R^o_g}
\subsection {Sufficiency Part}
We refer to Fig. \ref{fig:RR_C_RS} to determine the effect of varying $\lambda$ on $\mathcal{R}^o(Z)$ where $Z\in\mathcal{G}(P_1,P_2,P_s)$.
\begin{itemize}
	\item \emph{Point A:}
	
	\begin{equation*}
	R_p^A=\rho_p^o=\tau{\left(\frac{g_1^pP_1+g_2^pP_2}{g_s^p\lambda P_s+N_0}\right)}
	\end{equation*}
	\item \emph{Point D:}

	\begin{equation*}
	R_s^D= \rho_s^o=\tau{\left(\frac{P_sg_s^s}{g_1^sP_1+g_2^sP_2+N_0}\right)}
	\end{equation*}

  \item \emph{$R_s+R_p$:}

  \begin{equation*}
  \rho_{sp}^o= \tau{\left(\frac{g_1^pP_1+g_2^pP_2}{g_s^p\lambda P_s+N_0}\right)}
  +\tau\left(\frac{g_s^sP_s}{P_1g_1^s+P_2g_2^s+N_0}\right).
  \end{equation*}
\end{itemize}
It is obvious that if \eqref{eq:cond2} is satisfied, then $\rho_p^o$, $\rho_s^o$ and $\rho_{sp}^o$ increase as $\lambda$ decreases. Consequently, $\mathcal{R}^o(Z)$ at $\lambda=0$ includes all other $\mathcal{R}^o(Z)$ obtained at $0<\lambda\leq 1$. Hence, $\mathcal{R}^o(Z)$ coincides on $\mathcal{R}^o_g$ at $\lambda=0$.  

\subsection{Necessity Part}
Here we prove that the condition in \eqref{eq:cond2} is necessary for $\mathcal{R}^o(Z)$ to coincide on $\mathcal{R}^o_g$ at $\lambda=0$ and $Z\in\mathcal{G}(P_1,P_2,P_s)$. We do so by showing that, if \eqref{eq:cond2} is not satisfied, then for any two different values of $\lambda$ the corresponding rate regions do not contain one another. Assume that \eqref{eq:cond2} is not satisfied, then by referring to Fig. \ref{fig:RR_C_RS} we have:

\begin{itemize}
	\item \emph{Point A:}
	\begin{equation*}
	R_p^A=\tau\left(\frac{g_1^pP_1+g_2^pP_2}{g_s^p\lambda P_s+N_0}\right)
	\end{equation*}
	i.e., the $R_p^A$ decreases as $\lambda$ increases.
	\item \emph{Point D:}
	\begin{equation*}
	R_s^D=\tau\left({\frac{g_s^s\lambda P_s}{g_1^sP_1+g_2^sP_2+N_0}}\right) +\tau\left(\frac{g_s^p\bar{\lambda}P_s}{g_s^p\lambda P_s+N_0}\right)
	\end{equation*}
	then by substituting with $\bar{\lambda}=1-\lambda$ and differentiating $R_s^D$ with respect to $\lambda$ we get,
	\begin{equation}
	\label{eq:derivative1}	\frac{\partial{R_s^D}}{\partial{\lambda}}=\frac{\frac{1}{2\ln(2)}P_s(g_s^sN_0-g_s^p(P_1g_1^s+P_2g_2^s+N_0))}{(\lambda P_sg_s^p+N_0)(P_1g_1^s+P_2g_2^s+\lambda P_sg_s^s+N_0)}
\end{equation}
and since the condition \eqref{eq:cond2} is not satisfied, the numerator of \eqref{eq:derivative1} is always positive, therefore, $R_s^D$ increases as $\lambda$ increases.
\end{itemize}

Since $R_p^A$ decreases and $R_s^D$ increases as $\lambda$ increases, then for any two different values of $\lambda$ the corresponding rate regions will never contain one another. Hence the overall rate region $\mathcal{R}_g^o$ does not coincide on a certain $\mathcal{R}^o(Z)$ at a certain $\lambda$.
This concludes the proof.

\section{Proof of Theorem \ref{th:R^r_g}}
\label{App:R^r_g}
For the proof, we consider $i=1$, i.e., the secondary user is assumed to be able to decode the signal of primary user $1$.
\subsection{Sufficiency part}
In this part we show that, if inequality \eqref{eq:cond3} is satisfied then $\mathcal{R}_{1g}^r$ coincides on $\mathcal{R}^r_1(Z)$ at $\lambda=0$. We refer to Fig. \ref{fig:RR_CRSp} and determine the effect of varying $\lambda$ on $\mathcal{R}^r_1(Z)$, $Z\in\mathcal{G}(P_1,P_2,P_s)$ as follows. 
\subsubsection{At Point A}
\begin{equation*}
\begin{split}
R_p^{rA} =& \tau \left(\frac{g_2^pP_2}{g_s^p\lambda{P_s}+N_0}\right)+\min\left\{\tau\left(\frac{g_1^sP_1}{g_2^s P_2+N_0}\right) \right. ,\\
          & \left. \tau\left(\frac{g_1^p P_1}{g_s^p\lambda P_s+g_2^p P_2+N_0}\right) \right\}.
\end{split}
\end{equation*}
Therefore, $R_s^{rA}$ increases as $\lambda$ decreases.

\subsubsection{At Point F}
\begin{equation*}
\begin{split}
R_s^{rF}=\tau\left(\frac{g_s^s P_s}{g_2^s P_2+N_0}\right).
\end{split}
\end{equation*}
Hence, $R_s^{rF}$ does not depend on $\lambda$.
\subsubsection{$R_s^r+R_p^r=\rho_{sp}^r$}
\footnotesize
\begin{equation}
\label{eq:rho_sp^rg}
\begin{split}
\rho_{sp}^r=& \tau\left(\frac{g_s^s\lambda{P_s}}{g_2^sP_2+N_0}\right)+\tau\left(\frac{g_2^p P_2}{g_s^p\lambda{P_s}+N_0}\right)+\min \biggl\{  \\
            &  \overbrace{\tau\left(\frac{g_s^p\bar{\lambda}P_s+g_1^pP_1}{g_s^p\lambda P_s+g_2^pP_2+N_0}\right)}^{\mu_1},\overbrace{\tau\left(\frac{g_s^s\bar{\lambda}P_s+g_1^sP_1}{g_s^s\lambda{P_s}+g_2^sP_2+N_0}\right)}^{\mu_2} \\
            & \overbrace{\tau\left(\frac{g_s^p\bar{\lambda}P_s}{g_s^p\lambda{P_s}+g_2^pP_2+N_0}\right)+\tau\left(\frac{g_1^sP_1}{g_s^s\lambda{P_s}g_2^s+P_2+N_0}\right)}^{\mu_3},\\
            & \overbrace{\tau\left(\frac{g_1^pP_1}{g_s^p\lambda{P_s}+g_2^pP_2+N_0}\right)+ \tau\left(\frac{g_s^s\bar{\lambda}P_s}{g_s^s\lambda{P_s}+g_s^sP_2+N_0}\right)}^{\mu_4}\biggr\}
\end{split}
\end{equation}
\normalsize
\paragraph{When $\mu_1=\min\{\mu_1,\mu_2,\mu_3,\mu_4\}$ in \eqref{eq:rho_sp^rg}}
$$\rho_{sp}^r= \tau\left(\frac{g_s^s\lambda{P_s}}{g_2^sP_2+N_0}\right)+ \tau\left(\frac{g_s^p\bar{\lambda}P_s+g_1^pP_1+g_2^pP_2}{g_s^p\lambda{P_s}+N_0}\right).$$
\begin{equation*}
\begin{split}
\frac{\partial \rho_{sp}^r}{\partial \lambda}&=-\frac{0.5P_s(g_s^pg_2^sP_2+g_s^pN_0-g_s^sN_0)}{\ln 2(g_s^p\lambda{P_s}+N_0)(g_s^s\lambda P_s+g_2^sP_2+N_0)}\\
       &\leq 0 \quad \text{from \eqref{eq:cond3}}.
\end{split}
\end{equation*}
Hence, $\rho_{sp}^r$ decreases with $\lambda$. Note that, $\bar{\lambda}=1-\lambda$.
\paragraph{When $\mu_2=\min\{\mu_1,\mu_2,\mu_3,\mu_4\}$ in \eqref{eq:rho_sp^rg}}
$$\tau\left(\frac{g_s^sP_s+g_1^sP_1}{g_2^sP_2+N_0}\right)+\tau\left(\frac{g_2^pP_2}{g_s^p\lambda{P_s}+N_0}\right),$$
i.e., $\rho_{sp}^r$ decreases with $\lambda$.
\paragraph{When $\mu_3=\min\{\mu_1,\mu_2,\mu_3,\mu_4\}$ in \eqref{eq:rho_sp^rg}}
$$\rho_{sp}^r=\tau\left(\frac{g_s^s\lambda P_s+g_1^sP_1}{g_s^sP_2+N_0}\right)+\tau\left(\frac{g_s^p\bar{\lambda}P_s+g_2^pP_2}{g_s^p\lambda P_s+N_0}\right).$$
\begin{equation*}
\begin{split}
\frac{\partial \rho_{sp}^r}{\partial \lambda}&=-\frac{0.5P_s(g_s^pg_2^sP_2+g_s^pg_1^sP_1+g_s^pN_0-g_s^sN_0)}{\ln 2(g_s^p\lambda{P_s}+N_0)(g_s^s\lambda P_s+g_2^sP_2+g_1^sP_1+N_0)}\\
       &\leq 0 \quad \text{from \eqref{eq:cond3}}.
\end{split}
\end{equation*}
Thus, $\rho_{sp}$ decreases with $\lambda$.
\paragraph{When $\mu_4=\min\{\mu_1,\mu_2,\mu_3,\mu_4\}$ in \eqref{eq:rho_sp^rg}}
$$\rho_{sp}^r=\tau\left(\frac{g_s^sP_s}{g_s^sP_2+N_0}\right)+\tau\left(\frac{g_1^pP_1+g_2^pP_2}{g_s^p\lambda P_s+N_0}\right).$$
Therefore, $\rho_{sp}^r$ decreases with $\lambda$.

\subsubsection{$R_s^r+2R_p^r=\rho_{s2}^r$}
\footnotesize
\begin{equation*}
\begin{split}
\rho_{s2}^r=&2\tau\left(\frac{g_2^pP_2}{g_s^p\lambda P_s+N_0}\right)+2\sigma_p^*+\tau\left(\frac{g_s^s\lambda P_s}{g_2^sP_2+N_0}\right)\\
            &-\left[\sigma_p^*-\tau\left(\frac{g_1^sP_1}{g_s^s\lambda P_s+g_2^sP_2+N_0}\right)\right]^++\min\biggl\{\\
            &\tau\left(\frac{g_s^p\bar\lambda P_s}{g_s^p\lambda P_s+g_2^pP_2+N_0}\right), \tau\left(\frac{g_s^p\bar{\lambda}P_s+g_1^pP_1}{g_s^p\lambda{P_s}+g_2^pP_2+N_0}\right)-\sigma_p^*,\\
            &\tau\left(\frac{g_s^s\bar\lambda P_s}{g_s^s\lambda P_s+g_2^s P_2+N_0}\right), \tau\left(\frac{g_s^s\bar\lambda P_s}{g_s^s\lambda P_s+g_1^pP_1+g_2^pP_2+N_0}\right)\\
            &+\left[\tau\left(\frac{g_1^sP_1}{g_s^s\lambda P_s+g_2^sP_2+N_0}\right)-\sigma_p^*\right]^+\biggr\},\\
\sigma_p^*=&\min\left\{\tau\left(\frac{g_1^pP_1}{g_s^p\lambda P_s+g_2^pP_2+N_0}\right),\tau\left(\frac{g_1^sP_1}{g_2^sP_2+N_0}\right)\right\}.            
\end{split}
\end{equation*}
\normalsize
\paragraph{At $\sigma_p^*=\tau\left(\frac{g_1^pP_1}{g_s^p\lambda P_s+g_2^pP_2+N_0}\right)\leq \tau\left(\frac{g_1^sP_1}{g_2^sP_2+N_0}\right)$}
\footnotesize
\begin{equation*}
\begin{split}
\rho_{s2}^r=& 2\tau\left(\frac{g_1^pP_1+g_2^pP_2}{g_s^p\lambda P_s+N_0}\right)+\tau\left(\frac{g_s^s\lambda P_s}{g_2^sP_2+N_0}\right)\\
            &-\left[\tau\left(\frac{g_1^pP_1}{g_s^p\lambda P_s+g_2^pP_2+N_0}\right)-\tau\left(\frac{g_1^sP_1}{g_s^s\lambda P_s+g_2^sP_2+N_0}\right)\right]^+\\
            &+\min\biggl\{\tau \left(\frac{g_s^p\bar{\lambda}P_s}{g_s^p\lambda P_s+g_1^pP_1+g_2^pP_2+N_0}\right),\\
            &\tau\left(\frac{g_s^s\bar{\lambda} P_s}{g_s^s\lambda P_s+g_1^sP_1+g_2^sP_2+N_0}\right)+\biggl[\tau\left(\frac{g_1^sP_1}{g_s^s\lambda P_s+g_2^sP_2+N_0}\right)\\
            &-\tau \left(\frac{g_1^pP_1}{g_s^p\lambda P_s+g_2^pP_2+N_0}\right)\biggr]^+, \tau\left(\frac{g_s^s\bar{\lambda} P_s}{g_s^s\lambda P_s+g_2^sP_2+N_0}\right)\biggr\}.
\end{split}
\end{equation*}
\normalsize
\begin{itemize}
	\item If $\tau\left(\frac{g_1^pP_1}{g_s^p\lambda P_s+g_2^pP_2+N_0}\right)\leq\tau\left(\frac{g_1^sP_1}{g_s^s\lambda P_s+g_2^sP_2+N_0}\right)$
\end{itemize}
\footnotesize
\begin{equation}
\label{eq:rho_s2^rg1}
\begin{split}
\rho_{s2}^r=& 2\tau\left(\frac{g_1^pP_1+g_2^pP_2}{g_s^p\lambda P_s+N_0}\right)+\tau\left(\frac{g_s^s\lambda P_s}{g_2^sP_2+N_0}\right)+\min\biggl\{\\
            & \overbrace{\tau \left(\frac{g_s^p\bar{\lambda}P_s}{g_s^p\lambda P_s+g_1^pP_1+g_2^pP_2+N_0}\right)}^{\mu_5},\overbrace{\tau\left(\frac{g_s^s\bar{\lambda} P_s}{g_s^s\lambda P_s+g_2^sP_2+N_0}\right)}^{\mu_6},\\
            &\overbrace{\tau\left(\frac{g_s^s\bar{\lambda}P_s+g_1^sP_1}{g_s^s\lambda P_s+g_2^sP_2+N_0}\right)-\tau \left(\frac{g_1^pP_1}{g_s^p\lambda P_s+g_2^pP_2+N_0}\right)}^{\mu_7}\biggr\}.
\end{split}
\end{equation}
\normalsize

When $\mu_5=\min\{\mu_5,\mu_6,\mu_7\}$ in \eqref{eq:rho_s2^rg1} we have
\begin{equation}
\label{eq:rho_s2^rg12}
\begin{split}
\rho_{s2}^r=&\tau\left(\frac{g_s^p\bar\lambda P_s+g_1^pP_1+g_2pP_2}{g_s^p\lambda P_s+N_0}\right)+\tau\left(\frac{g_s^s\lambda P_s}{g_2^sP_2+N_0}\right)\\
            & +\tau \left(\frac{g_1^pP_1+g_2^pP_2}{g_s^p\lambda P_s+N_0}\right).
\end{split}
\end{equation}
Note that, the third term in \eqref{eq:rho_s2^rg12} is decreasing with $\lambda$, and the first derivative of the first two terms with respect to $\lambda$ is given by,
\begin{equation*}
\begin{split}
&-\frac{0.5P_s(g_s^pg_2^sP_2+g_s^pN_0-g_s^sN_0)}{\ln 2(g_s^p\lambda P_s+N_0)(g_s^s\lambda P_s+g_2^sP_2+N_0)}\\
&-\frac{0.5g_s^pP_s(g_2^pP_2+g_1^pP_1)}{\ln 2(g_s^p\lambda P_s+N_0)(g_s^p\lambda P_s+g_1^pP_1+g_2^pP_2+N_0)}.
\end{split}
\end{equation*}
Since inequality \eqref{eq:cond3} is satisfied for user $1$, then the derivative is negative and consequently $\rho_{s2}^r$ is decreasing with $\lambda$.

When $\mu_6=\min\{\mu_5,\mu_6,\mu_7\}$ in \eqref{eq:rho_s2^rg1}, we have
$$\rho_{s2}^r=2\tau \left(\frac{g_1^pP_1+g_2^pP_2}{g_s^p\lambda P_s+N_0}\right)+\tau\left(\frac{g_s^sP_s}{g_2^sP_2+N_0}\right),$$
i.e., $\rho_{s2}^r$ is decreasing with $\lambda$.

When $\mu_7=\min\{\mu_5,\mu_6,\mu_7\}$ in \eqref{eq:rho_s2^rg1}, we have
\begin{equation*}
\begin{split}
\rho_{s2}^r=&2\tau \left(\frac{g_1^pP_1+g_2^pP_2}{g_s^p\lambda P_s+N_0}\right)-\tau \left(\frac{g_1^pP_1}{g_s^p\lambda P_s+g_2^pP_2+N_0}\right)\\
            &+\tau\left(\frac{g_s^sP_s+g_1^sP_1}{g_2^sP_2+N_0}\right).
\end{split}
\end{equation*}
Hence, $\rho_{s2}^r$ is decreasing with $\lambda$.

\begin{itemize}
	\item If $\tau\left(\frac{g_1^sP_1}{g_s^s\lambda P_s+g_2^sP_2+N_0}\right)\leq\tau\left(\frac{g_1^pP_1}{g_s^p\lambda P_s+g_2^pP_2+N_0}\right)$
\end{itemize}

\footnotesize
\begin{equation}
\label{eq:rho_s2^rg2}
\begin{split}
\rho_{s2}^r=&2\tau \left(\frac{g_1^pP_1+g_2^pP_2}{g_s^p\lambda P_s+N_0}\right)-\tau \left(\frac{g_1^pP_1}{g_s^p\lambda P_s+g_2^pP_2+N_0}\right)\\
            &+\min\biggl\{\overbrace{\tau \left(\frac{g_s^p\bar{\lambda}P_s}{g_s^p\lambda P_s+g_1^pP_1+g_2^pP_2+N_0}\right)}^{\mu_5},\\
            &\overbrace{\tau\left(\frac{g_s^s\bar{\lambda} P_s}{g_s^s\lambda P_s+g_1^sP_1+g_2^sP_2+N_0}\right)}^{\mu_8}\biggr\}+\tau\left(\frac{g_s^s\lambda P_s+g_1^sP_1}{g_2^sP_2+N_0}\right)
\end{split}
\end{equation}
\normalsize

When $\mu_5=\min\{\mu_5,\mu_8\}$ in \eqref{eq:rho_s2^rg2}, then
\begin{equation}
\label{eq:rho_s2^rg22}
\begin{split}
\rho_{s2}^r=&\tau \left(\frac{g_1^pP_1+g_2^pP_2}{g_s^p\lambda P_s+N_0}\right)-\tau \left(\frac{g_1^pP_1}{g_s^p\lambda P_s+g_2^pP_2+N_0}\right)\\
            &+\tau \left(\frac{g_1^pP_1+g_2^pP_2}{g_s^p\lambda P_s+N_0}\right)+\tau\left(\frac{g_s^s\lambda P_s+g_1^sP_1}{g_2^sP_2+N_0}\right)\\
            &+\tau \left(\frac{g_s^p\bar{\lambda}P_s}{g_s^p\lambda P_s+g_1^pP_1+g_2^pP_2+N_0}\right).
\end{split}
\end{equation}
For all values of $0\leq\lambda\leq 1$, the difference between the first two terms in \eqref{eq:rho_s2^rg22} is always positive and decreasing as $\lambda$ increases. The first derivative of the last three terms in \eqref{eq:rho_s2^rg22} with respect to $\lambda$ is given by,
\begin{equation*}
\begin{split}
&-\frac{0.5P_s(g_s^pg_2^sP_2+g_s^pg_1^sP_1+g_s^pN_0-g_s^sN_0)}{\ln 2(g_s^p\lambda P_s+N_0)(g_s^s\lambda P_s+g_2^s P_2+g_1^s P_1+N_0)}\\
&\leq 0 \quad \text{from \eqref{eq:cond3}}.
\end{split}
\end{equation*}
Therefore, $\rho_{s2}^r$ is decreasing with $\lambda$.

When $\mu_8=\min\{\mu_5,\mu_8\}$ in \eqref{eq:rho_s2^rg2}, then
\begin{equation*}
\begin{split}
\rho_{s2}^r=&2\tau \left(\frac{g_1^pP_1+g_2^pP_2}{g_s^p\lambda P_s+N_0}\right)-\tau \left(\frac{g_1^pP_1}{g_s^p\lambda P_s+g_2^pP_2+N_0}\right)\\
            &+\tau\left(\frac{g_s^sP_s+g_1^sP_1}{g_2^sP_2+N_0}\right).
\end{split}
\end{equation*}
In the above formula, the difference between the first two terms is always positive and decreasing as $\lambda$ increases. The third term does not depend on $\lambda$. Hence, $\rho_{s2}^r$ is decreasing with $\lambda$.

\paragraph{At $\sigma_p^*=\tau\left(\frac{g_1^sP_1}{g_2^sP_2+N_0}\right)\leq\tau \left(\frac{g_1^pP_1}{g_s^p\lambda P_s+g_2^pP_2+N_0}\right)$}

\footnotesize
\begin{equation}
\label{eq:rho_s2^rg3}
\begin{split}
\rho_{s2}^r=&2\tau \left(\frac{g_2^pP_2}{g_s^p\lambda P_s+N_0}\right)+\tau\left(\frac{g_1^sP_1}{g_2^sP_2+N_0}\right)+\tau\left(\frac{g_s^s\lambda P_s+g_1^sP_1}{g_2^sP_2+N_0}\right)\\
            &+\min\biggl\{\overbrace{\tau\left(\frac{g_s^s\bar{\lambda} P_s}{g_s^s\lambda P_s+g_1^sP_1+g_2^sP_2+N_0}\right)}^{\mu_8},\\
            &\overbrace{\tau \left(\frac{g_s^p\bar{\lambda}P_s}{g_s^p\lambda P_s+g_2^p P_2+N_0}\right)}^{\mu_9},\\
            &\overbrace{\tau \left(\frac{g_s^p\bar{\lambda}P_s+g_1^p P_1}{g_s^p\lambda P_s+g_2^p P_2+N_0}\right)-\tau\left(\frac{g_1^sP_1}{g_2^sP_2+N_0}\right)}^{\mu_{10}}\biggr\}.
\end{split}
\end{equation}
\normalsize

When $\mu_8=\min\{\mu_8,\mu_9,\mu_{10}\}$ in \eqref{eq:rho_s2^rg3}, we have
\begin{equation*}
\begin{split}
\rho_{s2}^r=&2\tau \left(\frac{g_2^pP_2}{g_s^p\lambda P_s+N_0}\right)+\tau\left(\frac{g_1^sP_1}{g_2^sP_2+N_0}\right)\\
            &+\tau\left(\frac{g_s^sP_s+g_1^sP_1}{g_2^sP_2+N_0}\right).
\end{split}
\end{equation*}
That is, $\rho_{s2}^r$ is decreasing with $\lambda$.

When $\mu_9=\min\{\mu_8,\mu_9,\mu_{10}\}$ in \eqref{eq:rho_s2^rg3}, we have
\begin{equation}
\label{eq:rho_s2^rg32}
\begin{split}
\rho_{s2}^r=&\tau \left(\frac{g_2^pP_2}{g_s^p\lambda P_s+N_0}\right)+\tau\left(\frac{g_1^sP_1}{g_2^sP_2+N_0}\right)\\
            &+\tau\left(\frac{g_s^s\lambda P_s+g_1^sP_1}{g_2^sP_2+N_0}\right)+\tau \left(\frac{g_s^p\bar{\lambda}P_s+g_2^p P_2}{g_s^p\lambda P_s+N_0}\right).
\end{split}
\end{equation}
The first term in \eqref{eq:rho_s2^rg32} is decreasing with $\lambda$ for all values of $\lambda$. The first derivative of the other terms with respect to $\lambda$ is given by
\begin{equation*}
\begin{split}
&-\frac{0.5P_s(g_s^pg_s^sP_2+g_s^pg_1^sP_1+g_s^pN_0-g_s^sN_0)}{\ln 2(g_s^p\lambda P_s+N_0)(g_s^s\lambda P_s+g_2^sP_2+g_1^sP_1+N_0)}\\
&\leq 0 \quad \text{from \eqref{eq:cond3}}.
\end{split}
\end{equation*}
Hence, $\rho_{s2}^r$ is decreasing with $\lambda$.

When $\mu_{10}=\min\{\mu_8,\mu_9,\mu_{10}\}$ in \eqref{eq:rho_s2^rg3}, we have
\begin{equation}
\label{eq:rho_s2^rg33}
\begin{split}
\rho_{s2}^r=&\tau \left(\frac{g_2^pP_2}{g_s^p\lambda P_s+N_0}\right)+\tau \left(\frac{g_s^p\bar\lambda P_s+g_1^pP_1+g_2^pP_2}{g_s^p\lambda P_s+N_0}\right)\\
            &+\tau\left(\frac{g_s^s\lambda P_s+g_1^sP_1}{g_2^sP_2+N_0}\right).
\end{split}
\end{equation}
The first term in \eqref{eq:rho_s2^rg33} is decreasing with $\lambda$, and the first derivative of the other three terms with respect to $\lambda$ is given by,
\begin{equation*}
\begin{split}
&-\frac{0.5P_s(g_s^pg_2^sP_2+g_s^pg_1^sP_1+g_s^pN_0-g_s^sN_0)}{\ln 2(g_s^p\lambda P_s+N_0)(g_s^s\lambda P_s+g_1^sP_1+g_2^sP_2+N_0)}\\
&\leq 0 \quad \text{from \eqref{eq:cond3}}. 
\end{split}
\end{equation*}
Thus, $\rho_{s2}^r$ is decreasing with $\lambda$.

\subsubsection{$2R_s^r+R_p^r=\rho_{2p}^r$}
From \eqref{eq:cond3}, $$\sigma_s^*=\tau\left(\frac{g_s^s\bar{\lambda} P_s}{g_s^s\lambda P_s+g_2^sP_2+N_0}\right).$$
\footnotesize
\begin{equation*}
\begin{split}
\rho_{2p}^r=&2\tau\left(\frac{g_s^s P_s}{g_2^sP_2+N_0}\right)+\tau \left(\frac{g_2^pP_2}{g_s^p\lambda P_s+N_0}\right)\\
            &-\left[\tau\left(\frac{g_s^s\bar{\lambda} P_s}{g_s^s\lambda P_s+g_2^sP_2+N_0}\right)-\tau \left(\frac{g_s^p\bar{\lambda}P_s}{g_s^p\lambda P_s+g_2^p P_2+N_0}\right)\right]^+\\
            &+\min\biggl\{\tau\left(\frac{g_1^sP_1}{g_s^sP_s+g_2^sP_2+N_0}\right),\tau \left(\frac{g_1^p P_1}{g_s^p P_s+g_2^p P_2+N_0}\right)\\
            &+\left[\tau \left(\frac{g_s^p\bar{\lambda}P_s}{g_s^p\lambda P_s+g_2^p P_2+N_0}\right)-\tau\left(\frac{g_s^s\bar{\lambda} P_s}{g_s^s\lambda P_s+g_2^sP_2+N_0}\right)\right]^+,\\
            &\tau \left(\frac{g_1^pP_1}{g_s^p\lambda P_s+g_2^pP_2+N_0}\right)\biggr\}.
\end{split}
\end{equation*}
\normalsize
\paragraph{If $\tau\left(\frac{g_s^s\bar{\lambda} P_s}{g_s^s\lambda P_s+g_2^sP_2+N_0}\right)\leq \tau \left(\frac{g_s^p\bar{\lambda}P_s}{g_s^p\lambda P_s+g_2^p P_2+N_0}\right)$}
\footnotesize
\begin{equation}
\label{eq:rho_2p^rg1}
\begin{split}
\rho_{2p}^r=&2\tau\left(\frac{g_s^s P_s}{g_2^sP_2+N_0}\right)+\tau \left(\frac{g_2^pP_2}{g_s^p\lambda P_s+N_0}\right)\\
            &+\min\biggl\{\overbrace{\tau\left(\frac{g_1^sP_1}{g_s^s\lambda P_s+g_2^sP_2+N_0}\right)}^{\mu_{11}},\overbrace{\tau \left(\frac{g_1^pP_1}{g_s^p\lambda P_s+g_2^pP_2+N_0}\right)}^{\mu_{12}},\\
            &\overbrace{\tau \left(\frac{g_1^p P_1+g_s^p\bar{\lambda}P_s}{g_s^p\lambda P_s+g_2^p P_2+N_0}\right)-\tau\left(\frac{g_s^s\bar{\lambda} P_s}{g_s^s\lambda P_s+g_2^sP_2+N_0}\right)}^{\mu_{13}}\biggr\}.
\end{split}
\end{equation}
\normalsize
\begin{itemize}
	\item When $\mu_{11}=\min\{\mu_{11},\mu_{12},\mu_{13}\}$ in \eqref{eq:rho_2p^rg1}, then
\end{itemize}
\begin{equation*}
\begin{split}
\rho_{2p}^r=&2\tau\left(\frac{g_s^s P_s}{g_2^sP_2+N_0}\right)+\tau \left(\frac{g_2^pP_2}{g_s^p\lambda P_s+N_0}\right)\\
            &\tau\left(\frac{g_1^sP_1}{g_s^s\lambda P_s+g_2^sP_2+N_0}\right).
\end{split}
\end{equation*}
It is clear that, $\rho_{2p}^r$ is decreasing with $\lambda$.
\begin{itemize}
	\item When $\mu_{12}=\min\{\mu_{11},\mu_{12},\mu_{13}\}$ in \eqref{eq:rho_2p^rg1}, then
\end{itemize}
\begin{equation*}
\begin{split}
\rho_{2p}^r=&2\tau\left(\frac{g_s^s P_s}{g_2^sP_2+N_0}\right)+\tau \left(\frac{g_2^pP_2}{g_s^p\lambda P_s+N_0}\right)\\
            &+\tau \left(\frac{g_1^pP_1}{g_s^p\lambda P_s+g_2^pP_2+N_0}\right).
\end{split}
\end{equation*}
It is also clear that $\rho_{2p}^r$ is decreasing with $\lambda$.
\begin{itemize}
	\item When $\mu_{13}=\min\{\mu_{11},\mu_{12},\mu_{13}\}$ in \eqref{eq:rho_2p^rg1}, then
\end{itemize}
\begin{equation*}
\begin{split}
\rho_{2p}^r=&2\tau\left(\frac{g_s^s P_s}{g_2^sP_2+N_0}\right)+\tau \left(\frac{g_1^p P_1}{g_s^p P_s+g_2^p P_2+N_0}\right)\\
            &+\tau \left(\frac{g_s^p\bar\lambda P_s+g_2^pP_2}{g_s^p\lambda P_s+N_0}\right)-\tau\left(\frac{g_s^s\bar{\lambda} P_s}{g_s^s\lambda P_s+g_2^sP_2+N_0}\right).
\end{split}
\end{equation*}
\begin{equation*}
\begin{split}
\frac{\partial \rho_{2p}^r}{\partial \lambda}=&-\frac{0.5P_s(g_s^pg_2^sP_2+g_s^pN_0-g_s^sN_0)}{\ln 2(g_s^p\lambda P_s+N_0)(g_s^s\lambda P_s+g_2^s P_2+N_0)}\\
                                              &\leq 0 \quad\text{from \eqref{eq:cond3}}.
\end{split}
\end{equation*}
Thus, $\rho_{2p}^r$ is decreasing with $\lambda$.

\paragraph{If $\tau \left(\frac{g_s^p\bar{\lambda}P_s}{g_s^p\lambda P_s+g_2^p P_2+N_0}\right)\leq \tau\left(\frac{g_s^s\bar{\lambda} P_s}{g_s^s\lambda P_s+g_2^sP_2+N_0}\right)$}

\footnotesize
\begin{equation*}
\begin{split}
\rho_{2p}^r=&2\tau\left(\frac{g_s^s P_s}{g_2^sP_2+N_0}\right)+\tau \left(\frac{g_s^p\bar\lambda P_s+g_2^pP_2}{g_s^p\lambda P_s+N_0}\right)\\
            &-\tau\left(\frac{g_s^s\bar{\lambda} P_s}{g_s^s\lambda P_s+g_2^sP_2+N_0}\right)+\min\biggl\{\tau\left(\frac{g_1^sP_1}{g_s^sP_s+g_2^sP_2+N_0}\right),\\
            &\tau \left(\frac{g_1^p P_1}{g_s^p P_s+g_2^p P_2+N_0}\right)\biggr\}.
\end{split}
\end{equation*}
\normalsize
\begin{equation*}
\begin{split}
\frac{\partial \rho_{2p}^r}{\partial \lambda}=&-\frac{0.5P_s(g_s^pg_2^sP_2+g_s^pN_0-g_s^sN_0)}{\ln 2(g_s^p\lambda P_s+N_0)(g_s^s\lambda P_s+g_2^s P_2+N_0)}\\
                                              &\leq 0 \quad\text{from \eqref{eq:cond3}}.
\end{split}
\end{equation*}
Therefore, $\rho_{2p}^r$ is decreasing with $\lambda$.

Thus, since we showed that if \eqref{eq:cond3} is satisfied, assuming that the secondary receiver can decode the signal of primary user $1$, then $\rho_{p}^r$, $\rho_{sp}^r$, $\rho_{s2}^r$ and $\rho_{2p}^r$ decrease with $\lambda$, whereas $\rho_{s}^r$ does not depend on $\lambda$, hence, $\mathcal{R}^r_1(Z)$ at $\lambda=0$ coincides on $\mathcal{R}^r_{1g}$. And for any $\lambda_1$ and $\lambda_2$ such that $\lambda_1>\lambda_2$, $\mathcal{R}^r_1(Z)$ at $\lambda_1$ is a subset of $\mathcal{R}^r_1(Z)$ at $\lambda_2$.

\subsection{Necessity Part}
In this part of the proof we show that, if condition \eqref{eq:cond3} is not satisfied then $\mathcal{R}^r_{1g}$ does not coincide on any $\mathcal{R}^r_1(Z)$ for all values of $\lambda$. So, assume that \eqref{eq:cond3} is not satisfied, i.e., \begin{equation}
\label{eq:cond3_n}
N_0g_s^s> g_s^pg_2^sP_2+g_s^pN_0.
\end{equation}
By referring to Fig. \ref{fig:RR_CRSp}, the effect of $\lambda$ on $\mathcal{R}^r_1(Z)$ at points $A$ and $F$ is determined as follows.
\subsubsection{At Point A}
\begin{equation*}
\begin{split}
R_p^{rA}=&\tau \left(\frac{g_2^pP_2}{g_s^p\lambda P_s+N_0}\right)+\min\biggl\{\tau\left(\frac{g_1^sP_1}{g_2^sP_2+N_0}\right),\\
         &\tau \left(\frac{g_1^pP_1}{g_s^p\lambda P_s+g_2^pP_2+N_0}\right)\biggr\}.
\end{split}
\end{equation*}
It is clear that $R_p^{rA}$ is decreasing with $\lambda$.

\subsubsection{At Point F}

\begin{equation*}
\begin{split}
R_s^{rF}=\tau\left(\frac{g_s^s\lambda P_s}{g_2^sP_2+N_0}\right)+\tau \left(\frac{g_s^p\bar{\lambda}P_s}{g_s^p\lambda P_s+N_0}\right).
\end{split}
\end{equation*}
\begin{equation*}
\begin{split}
\frac{\partial R_s^{rF}}{\partial\lambda}&=\frac{0.5P_s(g_s^sN_0-(g_s^pg_s^sP_2+g_s^pN_0))}{\ln 2(g_s^p\lambda P_s+N_0)(g_s^s\lambda P_s+g_2^s P_2+N_0)}\\
                                         &>0\quad\text{from \eqref{eq:cond3_n}}.
\end{split}
\end{equation*}
Consequently, $R_s^{rF}$ is increasing with $\lambda$.

So, for any two different values of $\lambda$, the corresponding rate regions $\mathcal{R}^r_1(Z)$ do not include one another, thus $\mathcal{R}^r_{1g}$ does not coincide on $\mathcal{R}^r_1(Z)$ at any value of $\lambda$.

\section{Proof of Theorem \ref{th:noextention}}
\label{App:nxt}
From the definition of $\delta'^o(Z)$ and $\delta_1'^r(Z)$ it is clear that $\delta^o(Z)\subseteq\delta'^o(Z)$ and $\delta_1^r(Z)\subseteq\delta_1'^r(Z)$. Consequently, $\mathcal{R}^o(Z)\subseteq\mathcal{R}'^o(Z)$, $\mathcal{R}_1^r\subseteq\mathcal{R}_1'^r(Z)$ and $\mathcal{R}_1(Z)\subseteq\mathcal{R}_1'(Z)$. However, we show that if there exists $Z\in\mathcal{P}^*$ such that a rate tuple $(R_s,R_p)$ belongs to $\mathcal{R}_1'(Z)$ but does not belong to $\mathcal{R}_1(Z)$, then there exists another $Z'\in\mathcal{P}^*$ for which $(R_s,R_p)$ belongs to $\mathcal{R}_1(Z')$.

Following a similar procedure to that used in the proof of Theorem \ref{th:R^o(Z)}, the region $\mathcal{R}'^o(Z)$ is defined by:
\begin{equation}
\label{eq:Rp'}
R_p\leq I(Y_p;X_1X_2|WQ),
\end{equation}
\begin{equation}
\label{eq:R_s}
\begin{split}
R_s\leq & I(Y_s;U|WQ)+\min\{I(Y_s;W|Q),\\
        &I(Y_p;WX_1|X_2Q),I(Y_p;WX_2|X_1Q)\},
\end{split}
\end{equation}
\begin{equation}
\label{eq:Rs'+Rp'}
\begin{split}
R_s+R_p\leq & I(Y_s;U|WQ)+I(Y_p;X_1X_2|WQ)+\\
            &\min\{I(Y_s;W|Q),I(Y_p;W|Q)\}.
\end{split}
\end{equation}

\subsection{For $\mathcal{R}^o(Z)$}
\label{subsec:R^o(Z)}
Suppose that at a certain $Z\in\mathcal{P}^*$, $R_s'>I(Y_s;U|WQ)+I(Y_p;W|X_1X_2Q)$, hence, the rate tuple $(R_s',R_p')\in\mathcal{R}'^o(Z)$ but $(R_s',R_p')\notin\mathcal{R}^o(Z)$. From \eqref{eq:Rp'}-\eqref{eq:Rs'+Rp'}, $(R_s',R_p')$ has to satisfy
\begin{eqnarray}
\label{eq:Rs'e}
R_s\leq I(Y_s;UW|Q)=I(Y_s;X_s|Q),\\
R_p<I(Y_p;X_1X_2|Q).
\end{eqnarray}

Now, assume another $Z'\in\mathcal{P}^*$ such that $W=\phi$, i.e., no rate-splitting. At this $Z'$, $\mathcal{R}^o(Z')$ is given by
\begin{eqnarray}
\label{eq:R_pnrs}
R_s\leq I(Y_s;X_s|Q),\\
R_p\leq I(Y_p;X_1X_2|Q).
\end{eqnarray}

Then it is clear that $(R_s',R_p')\in\mathcal{R}^o(Z')$. Thus, $$\mathcal{R}'^o(Z)\subseteq\mathcal{R}^o(Z)\cup\mathcal{R}^o(Z').$$

\subsection {For $\mathcal{R}_1'^r(Z)$}

First, for a point $(R_s'',R_p'')$ such that $R_s''>I(Y_s;U|WQ)+I(Y_p;W|X_1X_2Q)$ at a specific $Z\in\mathcal{P}^*$, a similar argument as in the above subsection (Subsection \ref{subsec:R^o(Z)}), or in Lemma 2 of \cite{Chong}, can show that there exists $Z''\in\mathcal{P}^*$ such that $(R_s'',R_p'')\in\mathcal{R}_1^r(Z'')$.

Second, for another point $(R_s^{**},R_p^{**})$ such that $R_p^{**}>I(Y_p;X_2|WX_1Q)+I(Y_s;X_1|UWQ)$, or in other words $R_1^{**}>I(Y_s;X_1|UWQ)$, in this case, $\delta_1'^r(Z)\subset\delta'^o(Z)$. And since $\mathcal{R}'^o(Z)$ is the set of $(R_s,R_p)$ corresponding to $\delta'^o(Z)$ for which $R_s=S+T$ and $R_p=R_1+R_2$, then $\mathcal{R}_1'^r(Z)\subset\mathcal{R}'^o(Z)$. Moreover, it has been shown in the above subsection (Subsection \ref{subsec:R^o(Z)}) that $\mathcal{R}'^o(Z)\subseteq\mathcal{R}^o(Z)\cup\mathcal{R}^o(Z')$. Therefore, $$\mathcal{R}_1'^r(Z)\subseteq\mathcal{R}_1^r(Z)\cup\mathcal{R}_1^r(Z'')\cup\mathcal{R}^o(Z)\cup\mathcal{R}^o(Z').$$ Consequently,$$\mathcal{R}_1'=\mathcal{R}_1.$$

\end{document}